\documentclass[a4paper,allowtoday,onecolumn,accepted=2026-05-27]{quantumarticle} 
\pdfoutput=1 

\usepackage{dsfont} 
\usepackage{microtype} 
\usepackage[numbers,sort&compress]{natbib} 
\usepackage[utf8]{inputenc} 
\usepackage[T1]{fontenc} 
\usepackage[british]{babel} 
\usepackage[dvipsnames,x11names]{xcolor} 
\usepackage{amsmath,amssymb,amsthm,bm,amsfonts,bbm,stmaryrd} 
\usepackage{mathtools} 
\usepackage{physics} 
\usepackage{multirow}
\usepackage{booktabs}
\usepackage{diagbox}
\usepackage[colorlinks=true,citecolor=green,linkcolor=Purple,urlcolor=Magenta]{hyperref} 
\usepackage[capitalize,nameinlink]{cleveref}
\usepackage{graphicx}
\usepackage{subcaption}
\usepackage{booktabs}
\usepackage{float}
\usepackage[scaled=0.8]{FiraMono}

\DeclareMathOperator{\Sp}{Sp}

\renewcommand{\leq}{\leqslant}

\renewcommand{\geq}{\geqslant}

\newcommand{\1}{\mbox{1}\hspace{-0.25em}\mbox{l}}
\newcommand{\id}{\mathds{1}}

\newcommand{\poly}{\operatorname{poly}}

\newcommand{\code}[1]{\texttt{\detokenize{#1}}}

\makeatletter
\newcommand{\doublewidetilde}[1]{{%
  \mathpalette\double@widetilde{#1}%
}}
\newcommand{\double@widetilde}[2]{%
  \sbox\z@{$\m@th#1\widetilde{#2}$}%
  \ht\z@=.9\ht\z@
  \widetilde{\box\z@}%
}

\makeatother

\newtheorem{theorem}{Theorem}

\begin{document}

\title{Measurement incompatibility and quantum steering via linear programming}

\author{\fontsize{11}{13}\selectfont Lucas E. A. Porto}
\address{Sorbonne Université, CNRS, LIP6, F-75005 Paris, France}
\address{Instituto de Física Gleb Wataghin, Universidade Estadual de Campinas (Unicamp), Rua Sérgio Buarque de Holanda 777, Campinas,
São Paulo 13083-859, Brazil}
\orcid{0000-0001-9509-5951}
\author{\fontsize{11}{13}\selectfont Sébastien Designolle}
\address{Zuse Institute Berlin, Takustraße 7, 14195 Berlin, Germany}
\address{Inria, ENS de Lyon, UCBL, LIP, 69342, Lyon Cedex 07, France}
\orcid{0000-0003-0303-3556}
\author{\fontsize{11}{13}\selectfont Sebastian Pokutta}
\address{Zuse Institute Berlin, Takustraße 7, 14195 Berlin, Germany}
\address{Berlin Institute of Technology, 10587 Berlin, Germany}
\orcid{0000-0001-7365-3000}
\author{\fontsize{11}{13}\selectfont Marco Túlio Quintino}
\address{Sorbonne Université, CNRS, LIP6, F-75005 Paris, France}
\orcid{0000-0003-1332-3477}
\date{10th June 2026}

\begin{abstract}
The problem of deciding whether a set of quantum measurements is jointly measurable is known to be equivalent to determining whether a quantum assemblage is unsteerable.
This problem can be formulated as a semidefinite program (SDP).
However, the number of variables and constraints in such a formulation grows exponentially with the number of measurements, rendering it intractable for large measurement sets.
In this work, we circumvent this problem by transforming the SDP into a hierarchy of linear programs that compute upper and lower bounds on the incompatibility robustness with a complexity that grows polynomially in the number of measurements.
The hierarchy is guaranteed to converge and it can be applied to arbitrary measurements --- including non-projective POVMs (Positive Operator-Valued Measures) --- in arbitrary dimensions.
While convergence becomes impractical in high dimensions, in the case of qubits our method reliably provides accurate upper and lower bounds for the incompatibility robustness of sets with several hundred measurements in a short time using a standard laptop.
We also apply our methods to qutrits, obtaining non-trivial upper and lower bounds in scenarios that are otherwise intractable using the standard SDP approach, although such bounds are significantly looser than the ones obtained in the qubit case.
Finally, we show how our methods can be used to construct local hidden state models for states (i.e., to prove that a state cannot lead to steering under any possible local measurements), or conversely, to certify that a given state exhibits steering; for two-qubit quantum states, our approach is comparable to, and in some cases outperforms, the current best methods.
\end{abstract}

\maketitle
\vspace{-20pt}
\tableofcontents

\section{Introduction}
One of the most intriguing features of quantum theory is its prediction of incompatible measurements, that is, measurements that cannot be jointly performed~\cite{Guhne2023IncompReview}.
Many of the counter-intuitive phenomena exhibited by quantum systems are strongly related to this fact.
For example, the use of incompatible measurements in a Bell scenario is a necessary condition to demonstrate nonlocality~\cite{bell64, brunner_review, fine82}, although not sufficient~\cite{quintino15Incompatible, hirsch2018, bene2018} (see however \cite{Plavala2024incomp} for a discussion on a multipartite scenario).
Similarly, the problem of deciding whether a given \textit{assemblage} is \textit{steerable} is mathematically equivalent to deciding whether a given set of measurements is incompatible~\cite{quintino14, uola14, Uola2020SteeringReview}.
Moreover, measurement incompatibility also plays a crucial role for several information-processing protocols~\cite{Buscemi2020PAM,Carmeli2019PAM, Guerini2019PAM, acin07}.
For instance, every set of incompatible measurements provides an advantage in a quantum state discrimination task~\cite{Paul2019PAM, carmeli2018}.

In this context, a prominent problem from both foundational and practical perspectives is that of understanding which sets of measurements can be jointly performed.
In textbook quantum theory, where typically only projective measurements are discussed, the notion of joint measurability is reduced to commutativity of observables~\cite{Schumacher2010textbook}.
Nevertheless, for general positive operator-valued measures (POVMs) the question becomes more subtle.
Except in some of the simplest cases (see, e.g., Refs.~\cite{busch1986, yu2008biased, busch2010biased, stano2010biased, pal2011, yu2013, grinko2024}), there is no analytical criterion that completely characterizes measurement incompatibility.

Numerically, the problem of deciding whether a given set of measurements is jointly measurable can be cast as a semidefinite program (SDP)~\cite{wolf09}.
In fact, more than just solving the decision problem, this formulation is also able to compute \textit{quantifiers} of incompatibility, which indicate how close the measurements are to being compatible~\cite{cavalcanti2016quantifiers, designolle2019quantifiers}.
However, the SDP can only be used in practice when dealing with a sufficiently small number of measurements, as the number of variables in the program quickly explodes.
More precisely, to analyze a set of $m$ measurements with $k$ outcomes using the SDP, one is required to list all the $k^m$ deterministic post-processing strategies, and to each of them associate one semidefinite variable.
Hence, the program's memory consumption grows exponentially with the number of measurements studied, which also results in exponential time complexity.

In this work, we approach the problem from a different perspective, proposing a method to characterize measurement incompatibility which scales efficiently with the number of considered measurements.
The method is a modification of the standard SDP approach which prevents the necessity of listing the deterministic post-processing strategies, thus eliminating the main source of the scalability problem.
In a nutshell, this is done by letting the post-processing strategies be variables of the optimization problem, while approximating the set of positive semidefinite operators of unit trace in the given dimension, i.e., the set of quantum states, by a previously chosen polytope.
The resulting optimization problem becomes a linear program (LP), that computes lower and upper bounds on incompatibility quantifiers of the studied set of measurements.
The tightness of these bounds is dictated by the quality of the chosen polytope, as quantified by its \textit{shrinking factor}.
Thus, the reduction on the complexity of the problem comes at the cost of a controlled loss in accuracy.

The method can be applied to arbitrary sets of measurements in arbitrary dimensions, and, in practice its main limitation is the aforementioned requirement of such polytope approximations.
On the one hand, in the case of qubits, since the set of states can be represented in the Bloch ball, a well studied and fairly simple set, our method is able to achieve a remarkably high precision for sets of several hundred measurements, while the standard SDP cannot go beyond sets of $\sim20$ measurements.
On the other hand, for higher dimensions it becomes ever more difficult to approximate the set of states using polytopes with a reasonable number of vertices~\cite{Fawzi18}.
We propose a few ideas to construct such polytopes, and as an illustration we use some of them to study the incompatibility of qutrit measurements.
The results obtained are decent, but significantly looser than the ones obtained in the qubit case.

Recently, another method for deciding measurement incompatibility that also avoids the standard SDP approach has been proposed~\cite{grinko2024}.
It consists of an analytical criterion that generalizes the ones proposed in Refs.~\cite{busch1986, pal2011, yu2013} for any finite number of unbiased qubit observables.
While Ref.~\cite{grinko2024} brings great theoretical insights for the qubit unbiased measurement problem, in computational terms this criterion leads to a second-order cone program whose performance is similar to the SDP approach.
Thus, the technique proposed in Ref.~\cite{grinko2024} could not be used to tackle instances which are not tractable by standard methods.
For instance, it cannot deal with qubit measurement sets with more than $\sim 20$ measurements, while, as detailed later, the methods presented here can tackle several hundred arbitrary qubit POVMs in seconds/minutes.

In addition to measurement incompatibility, the method we introduce also finds applications in the study of quantum steering.
Since the problem of deciding the steerability of an assemblage is equivalent to deciding the incompatibility of a set of measurements, the method can be straightforwardly used to decide whether a given assemblage is steerable.
More interestingly, the method is also useful to decide which \textit{states}, rather than assemblages, can produce steerable correlations.
We show that it is not difficult to use it in order to certify that a given state is steerable.
Furthermore, it is also possible to combine it with the methods developed in Refs.~\cite{hirsch15, cavalcanti15} in order to construct \textit{local hidden-state (LHS) models} for quantum states, i.e., to prove that states cannot exhibit steering under any possible local measurements.
For the case of two-qubit states, we find these approaches to be comparable, and in some cases better than the state-of-the-art method proposed in Ref.~\cite{Nguyen2019LHSCriterion}.

\section{Preliminaries}
Let us start by briefly reviewing the basic concepts related to our discussion.
More specifically, in this section we formally define the notion of measurement incompatibility and present an example of an incompatibility quantifier, showing how it can be computed via SDP.
Moreover, we introduce the concept of quantum steering and discuss the equivalence between the problem of deciding the steerability of an assemblage and the incompatibility of a set of measurements.

\subsection{Measurement incompatibility}
In quantum theory, measurements are represented by POVMs.
A set of linear operators $\{M_a\}_a$ acting on a complex Hilbert space of dimension $d$ is a POVM if all operators are positive semidefinite, $M_a \geq 0$, and if they add up to identity $\sum_a M_a = \1$, where $\1$ is the identity operator.
Here, the label $a \in [k]$ indicates the measurement outcome, and $[k]:=\{1,\ldots,k\}$.
A set of quantum measurements is then described by $\{M_{a|x}\}$, where $M_{a|x}$ denotes the element of measurement ${x \in [m]}$ associated with outcome $a \in [k]$.
Since for each fixed $x$ the set $\{M_{a|x}\}_a$ is a POVM, it holds that  $M_{a|x} \geq 0$ and $\sum_a M_{a|x} = \1$ for all $x$.
A set of measurements $\{M_{a|x}\}$ is \textit{compatible} (also referred to as \textit{jointly measurable}), if there exists a POVM $\{G_{\lambda}\}$ and probability distributions $p(a|x, \lambda)$ such that
\begin{equation}\label{eq:def_jm}
  M_{a|x} = \sum_\lambda p(a|x, \lambda)G_\lambda.
\end{equation}
That is, these measurements are compatible if their statistics can be retrieved from a single measurement $\{G_\lambda\}$, often referred to as \textit{joint} or \textit{parent} measurement, and some classical post-processing strategies, represented by the probabilities $p(a|x, \lambda)$.
If such a decomposition does not exist, the measurements $\{M_{a|x}\}$ are said to be \textit{incompatible}.
As standard in the field, the expression measurement compatibility and joint measurability are used interchangebly.

For any given set of measurements $\{M_{a|x}\}$, the problem of whether there exists a parent measurement and post-processing strategies satisfying \cref{eq:def_jm} can be readily written as an SDP.
Actually, this SDP formulation of the problem goes beyond that and also computes \textit{quantifiers} of measurement incompatibility.
Various notions of such quantifiers have been proposed in the literature~\cite{designolle2019quantifiers}, and most of them can be computed via SDP.
The main technique we propose in this work can be applied to any quantifier computable via SDP.
For simplicity, however, we develop our discussion in terms of the so-called \textit{incompatibility depolarizing robustness}, defined as
\begin{equation} \label{eq:depolarizing_ass}
  \eta^*(\{M_{a|x}\}) = \max\{\eta ~|~ \{\Lambda_\eta(M_{a|x})\}~ \text{is jointly measurable}\},
\end{equation}
where $\Lambda_\eta$ is the \textit{depolarizing map} with visibility $\eta \in [0,1]$, which acts on an operator $A$ as
\begin{equation}\label{eq:depolarizing_map}
  \Lambda_\eta (A) = \eta A + (1 - \eta) \Tr(A)\frac{\1}{d}.
\end{equation}

To compute the incompatibility depolarizing robustness of a given set of measurements $\{M_{a|x}\}$ with an SDP, first notice that in \cref{eq:def_jm} the probabilities $p(a|x, \lambda)$ can be taken to be deterministic~\cite{Guhne2023IncompReview, ali2009}.
That is, for the classical post-processing we may consider only the deterministic strategies $D(a | x, \lambda) = \delta_{a, f_\lambda(x)}$, where the function $f_\lambda : [m] \rightarrow [k]$ assigns a fixed outcome to every measurement, and $\lambda \in [k^m]$ labels all possible such functions.
Then, one can write
\begin{subequations} \label{eq:main_SDP}
  \begin{align}
    \eta^*(\{M_{a|x}\}) = \max \; &\eta \\
    \text{s.t.} \; &\eta M_{a|x} + (1 - \eta)\Tr (M_{a|x}) \frac{\1}{d} = \sum_{\lambda = 1}^{k^m}D(a|x, \lambda)G_\lambda, \quad \forall a \in [k],~x \in [m]\\
    &G_\lambda \geq 0, \quad \forall \lambda \in [k^m],
  \end{align}
\end{subequations}
which is an SDP where the variables are $\eta$ and $\{G_\lambda\}$.
If the solution to this problem is ${\eta(\{M_{a|x}\})\geq1}$, then the measurements $\{M_{a|x}\}$ are jointly measurable.
Otherwise, they are incompatible.

An important observation about the program above is the fact that it requires $O(k^m)$ semidefinite variables.
Thus, as one increases the number of measurements under study, one has SDPs with exponentially many more variables.
This is the main problem we tackle in this work.
However, before we discuss our proposed method, let us briefly revisit the basic ideas of quantum steering.

\subsection{Quantum steering}\label{subsec:steering_intro}
Suppose that two observers, Alice and Bob, each receive a quantum system prepared by a third party.
Then, Alice performs a measurement associated to the classical label $x \in [m]$ on her system, and obtains outcome $a \in [k]$.
She then communicates her measurement and outcome to Bob, who stores his state with the label $\rho_{a|x}$.
After repeating this procedure a sufficient amount of times, Bob is able to perform tomography and estimate each $\rho_{a|x}$, for all $x$ and $a$.
Additionally, he can also estimate the probabilities $p(a|x)$.
This information is concisely represented by the \textit{steering assemblage} $\{\sigma_{a|x}\}$, where $\sigma_{a|x} = p(a|x) \rho_{a|x}$.

In fact, a steering assemblage can be any set of operators $\{\sigma_{a|x}\}$ satisfying the constraints of positivity $\sigma_{a|x} \geq 0$ for all $x$ and $a$, and a normalization of the form $\sum_a \sigma_{a|x} = \rho_B$, where $\Tr(\rho_B) = 1$, for all $x$.
In terms of the protocol described above, $\rho_B$ represents Bob's reduced state without the knowledge of Alice's measurement and outcome.

When given an assemblage $\{\sigma_{a|x}\}$, if there exists a set of states $\{\rho_\lambda\}$ and probability distributions $p(\lambda)$ and $p(a|x, \lambda)$ such that
\begin{equation}\label{eq:def_unst}
  \sigma_{a|x} = \sum_\lambda p(\lambda)p(a|x, \lambda)\rho_\lambda
\end{equation}
for all $a \in [k]$ and $x \in [m]$, then we say that the assemblage is \textit{unsteerable}, or it is said to admit a \textit{local hidden-state} (LHS) model.
Otherwise, the assemblage is said to be \textit{steerable}.

Notably, if Alice and Bob share a maximally entangled state $\ket{\phi^+_d} = \frac{1}{\sqrt{d}}\sum_{i = 0}^{d-1}\ket{i}\ket{i}$ and Alice applies the measurements $\{M_{a|x}\}$ to her system, then the resulting assemblage is given by $\sigma_{a|x} = \frac{1}{d}M_{a|x}^T$.
This assemblage is unsteerable if, and only if, the measurements $\{M_{a|x}\}$ are jointly measurable~\cite{quintino14, uola14}.
That is, any joint measurability problem can be thought of as a steering problem.

Conversely, the problem of deciding whether an assemblage $\{\sigma_{a|x}\}$ is unsteerable can also be mapped to deciding whether a set of measurements is compatible.
For any given (non-signalling) assemblage $\{\sigma_{a|x}\}$, construct a set of measurements $\{M_{a \vert x}\}$ in the following way.
Take the state $\rho_B = \sum_a \sigma_{a \vert x}$, and consider the operator $\tilde{\rho}_B^{-1/2} = \sum_i \frac{1}{\sqrt{p_i}}\ketbra{\phi_i}$, where $p_i$ are the nonzero eigenvalues of $\rho_B$ and $\ket{\phi_i}$ their corresponding eigenvectors.
Then, define $M_{a \vert x} = \tilde{\rho}_B^{-1/2} \sigma_{a|x} \tilde{\rho}_B^{-1/2}$.
The positive operators $M_{a \vert x}$ are normalised to identity in the subspace formed by the image of $\rho_B$, i.e., $\sum_a M_{a \vert x} = \id_{\text{Im}(\rho_B)}$, and therefore constitute proper measurements in this vector space.
It then follows that the set $\{M_{a \vert x}\}$ is jointly measurable if and only if the assemblage $\{\sigma_{a|x}\}$ is unsteerable (see Ref.~\cite{uola15} for a detailed discussion).

In complete analogy to the previous subsection, given the steering assemblage $\{\sigma_{a|x}\}$ we define the \textit{steering depolarizing robustness} (see~\cref{eq:depolarizing_map}) as
\begin{equation}
  \eta^*(\{\sigma_{a|x}\}) = \max\{\eta ~|~ \{\Lambda_\eta(\sigma_{a|x})\}~ \text{is unsteerable}\},
\end{equation}
which can be computed by the SDP
\begin{subequations} \label{eq:SDP_steering}
  \begin{align}
    \eta^*(\{\sigma_{a|x}\}) = \max \; &\eta \\
    \text{s.t.} \; &\eta \sigma_{a|x} + (1 - \eta)\Tr (\sigma_{a|x}) \frac{\1}{d} = \sum_{\lambda = 1}^{k^m}D(a|x, \lambda)\sigma_\lambda, \quad \forall a \in [k], ~x \in [m] \label{eq:SDP_steering_lhs_constraint}\\
    &\sigma_\lambda \geq 0, \quad \forall \lambda \in [k^m],
  \end{align}
\end{subequations}
where $\sigma_\lambda$ are sub-normalised states, associated to the product $p(\lambda)\rho_\lambda$ in \cref{eq:def_unst}, and $D(a \vert x, \lambda)$ are deterministic probability distributions, defined as in \cref{eq:main_SDP}.
As expected, this program suffers from the same problem as the one discussed in the previous subsection.
That is, it contains a number of variables that grows exponentially with the number of considered measurements $m$.
This means that, in practice, this approach is limited to problems with sufficiently few measurements.

\section{Efficient method for quantifying steering/measurement incompatibility}
In this section, we present a method to quantify measurement incompatibility and steering which scales efficiently with the number of measurements.
For simplicity, it is presented in terms of steering rather than measurement incompatibility, but the reader should keep in mind that these two problems are equivalent, as discussed in the previous section.
The method consists of a linear program that computes lower and upper bounds on steering quantifiers of a given assemblage.
It requires a polytope approximation of the set of quantum states in the desired dimension, which has to be chosen before running the LP.
After presenting the core idea of the method, we analyze how the polytope approximations are related to the accuracy of the given bounds, which leads to a formalization of the computational complexity of the method.
Finally, we discuss how to construct polytopes to be used in the linear program in practice.

\subsection{Linear program for steerability}
Consider the problem of computing the depolarizing robustness of a given steering assemblage $\{\sigma_{a|x}\}$ in dimension $d$, where $a \in [k]$, $x \in [m]$.
If $m$ is not too large, one may compute it via the SDP \eqref{eq:SDP_steering}.
However, since the number of variables in the SDP grows exponentially with $m$, this approach quickly becomes impractical.
This is due to the fact that in order to write the problem as an SDP, one has to write down each one of the $k^m$ deterministic strategies $D(a|x, \lambda)$, and associate one positive semidefinite variable $\sigma_\lambda$ to each one of these strategies.
Indeed, comparing the right-hand side of \cref{eq:def_unst} with the right-hand side of \cref{eq:SDP_steering_lhs_constraint}, we see that $p(a|x, \lambda)$ was turned into $D(a|x, \lambda)$, and the product $p(\lambda)\rho_\lambda$ was merged into the sub-normalized states $\sigma_\lambda$, which are variables of the SDP.

\begin{figure}[ht]
  \centering
  \begin{subfigure}{0.35\textwidth}
    \includegraphics[page=1,width=\textwidth]{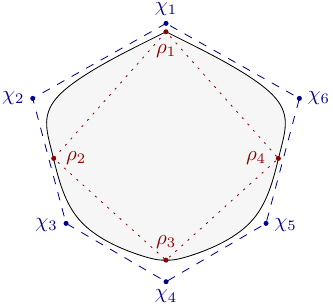}
    \caption{}
    \label{fig:out}
  \end{subfigure}
  \hspace{2cm}
  \begin{subfigure}{0.35\textwidth}
    \includegraphics[page=2,width=\textwidth]{figs/shrinking.pdf}
    \caption{}
    \label{fig:shr}
  \end{subfigure}
  \caption{
    Illustration of the inner and outer approximations of the state space (light grey).
    (a) The red dotted polytope with vertices $\{\rho_1, \rho_2, \rho_3, \rho_4\}$ is an inner approximation of the state space.
    The blue dashed polytope with vertices $\{\chi_1, \chi_2, \chi_3, \chi_4, \chi_5, \chi_6\}$ is an outer approximation of the state space.
    (b) The dark grey area represents the image of the state space under the action of $\Lambda_r$, which is contained in the polytope formed by the set of states $\{\rho_1, \rho_2, \rho_3, \rho_4\}$.
    This set's shrinking factor is the maximum $r$ for which such an inclusion holds.
    We call $\rho^*$ the worst state, that is, such that $\Lambda_r(\rho^*)$ lies exactly on the corresponding facet.
  }
  \label{fig:approx}
\end{figure}

Alternatively, now consider a fixed set of $d$-dimensional states $\{\rho_\lambda\}$ (see \cref{fig:out}), say, with $\lambda \in [n]$, and let us define an \textit{approximate} steering depolarizing robustness of the assemblage $\{\sigma_{a|x}\}$ with respect to the set of states $\{\rho_\lambda\}$ by
\begin{equation}\label{eq:def_approx_robustness}
  \tilde{\eta}(\{\sigma_{a|x}\}, \{\rho_\lambda\}) = \max\{\eta ~|~ \Lambda_\eta(\sigma_{a|x}) = \sum_{\lambda = 1}^n p(\lambda)p(a|x, \lambda)\rho_\lambda, ~\forall a \in [k],~x \in [m]\}.
\end{equation}
That is, $\tilde{\eta}(\{\sigma_{a|x}\}, \{\rho_\lambda\})$ can be thought of as the steering depolarizing robustness of $\{\sigma_{a|x}\}$ (see \cref{eq:depolarizing_ass}) if one restricts the LHS models to only include the fixed $\{\rho_\lambda\}$ as the hidden states.
This can be seen as if we were approximating the set of $d$-dimensional quantum states by the polytope formed by vertices $\{\rho_\lambda\}$.
Since each $\rho_\lambda$ is a proper state, it naturally follows that ${\tilde{\eta}(\{\sigma_{a|x}\}, \{\rho_\lambda\}) \leq \eta^*(\{\sigma_{a|x}\})}$.

Notice that the definition given by  \cref{eq:def_approx_robustness} still makes sense if we replace the set of states $\{\rho_\lambda\}$ by a set of $n$ general $d \times d$ matrices, not necessarily positive nor having unit trace.
Importantly, these sets of general matrices can be used to obtain upper bounds on $\eta(\{\sigma_{a \vert x}\})$.
Indeed, take a set of $d \times d$ matrices $\{\chi_\lambda\}$ such that the set of density matrices is contained in the convex hull of $\{\chi_\lambda\}$, i.e., such that every quantum state in dimension $d$ can be written as a convex combination of the matrices $\{\chi_\lambda\}$ (see \cref{fig:out}).
Then, it holds that $\tilde{\eta}(\{\sigma_{a|x}\}, \{\chi_\lambda\}) \geq \eta^*(\{\sigma_{a|x}\})$.

One advantage of defining an approximate steering robustness in such a way is the fact that, fixing a set of states (or general $d \times d$ matrices) $\{\rho_\lambda\}$, it can be computed by the linear program
\begin{subequations}\label{eq:main_LP}
  \begin{align}
    \tilde{\eta}(\{\sigma_{a|x}\}, \{\rho_\lambda\}) = \max \; &\eta \\
    \text{s.t.} \; &\eta \sigma_{a|x} + (1 - \eta)\Tr (\sigma_{a|x}) \frac{\1}{d} = \sum_{\lambda = 1}^n p(\lambda, a|x)\rho_\lambda, \quad \forall a \in [k],~x \in [m]\\
    &p(\lambda, a|x) \geq 0, \quad \forall \lambda \in [n],~a \in [k],~x \in [m]\\
    &\sum_{a = 1}^k p(\lambda, a|x) = \sum_{a = 1}^k p(\lambda, a|x'), \quad \forall \lambda \in [n],~x \in [m],~x' \in [m],
  \end{align}
\end{subequations}
whose variables are $\eta$ and the probability distributions $p(\lambda, a|x)$.
Notice that in terms of the right-hand side of \cref{eq:def_unst}, the latter are associated with the product $p(\lambda)p(a|x, \lambda)$.
The above program thus has $O(kmn)$ \textit{real} variables and $O(kmn)$ constraints, meaning its size scales only \textit{linearly} with the number of measurements $m$.
This contrasts with the exponential scaling of the standard SDP approach \eqref{eq:SDP_steering}.
Moreover, in the program \eqref{eq:main_LP} the variables are real numbers, whereas in the SDP \eqref{eq:SDP_steering} the variables are semidefinite matrices, inherently more complex objects.
In this sense, the linear program \eqref{eq:main_LP} provides a substantial improvement in terms of the scalability of the problem with the number of measurements.

Importantly, this improvement is not simply due to approximating the SDP constraints with an LP, which cannot be done efficiently in general; see Refs.~\cite{fmptw2011jour,bfps2012jour,braun2014matching}, where it is shown that approximating a single SDP constraint with an LP can require exponential size LPs.
The linear program \eqref{eq:main_LP} is thus a non-trivial improvement over the SDP \eqref{eq:SDP_steering} and much more akin to a column generation approach customary in linear programming, trading off accuracy for size and efficiency.

Indeed, as we will formalize in the next sections, the main downside of the LP is related to its accuracy.
It does not compute the exact steering robustness of the given assemblage, but rather lower and upper bounds on it.
One should notice, however, that in practice the SDP approach \cref{eq:SDP_steering} is also approximate, since a semidefinite program is only polynomial-time solvable in an approximate fashion \cite{grotschel1981ellipsoid, alizadeh1995interior, nesterov1994interior, vandenberghe1996semidefinite}.
Even so, the standard SDP scales better with the tolerance than the method we propose.
Therefore, the cost of being more efficient with respect to the number of measurements is a controlled loss in precision.

To better understand this trade off between accuracy and efficiency, notice that the precision of the bounds provided by the linear program \eqref{eq:main_LP} has to be strongly related to the set of states $\{\rho_\lambda\}$ (or arbitrary matrices $\{\chi_\lambda\}$), which one has to choose before running the program.
In fact, we will show that there exists a single parameter of the set of states $\{\rho_\lambda\}$, the so-called \textit{shrinking factor}, which quantifies the precision of the bounds computed by our method when such a set is used.
In the first place, this fact interestingly characterizes adequate polytopes for the method.
And furthermore, it also enables a formalization of its computational complexity.

\subsection{Understanding the approximation factor}
Consider a set of density matrices $\{\rho_\lambda\}$ acting on a $d$-dimensional Hilbert space and containing the maximally mixed state in the interior of its convex hull.
This set's \textit{shrinking factor} is defined as the maximum $r \in \mathbb{R}$ such that for any given $d$-dimensional state $\rho$, one can write $\Lambda_r(\rho)$ as a convex combination of the states $\{\rho_\lambda\}$ (see \cref{fig:shr}).
In other words, it is the maximum $r$ such that the image of the state space under the action of the depolarizing channel $\Lambda_r$ lies entirely inside the polytope generated by $\{\rho_\lambda\}$.
Thus, the shrinking factor of a finite set of states belongs to the interval $(0,1)$, and the closer it is to one, the better the convex hull of the set covers the whole state space.

In the case of qubits, where the state space is represented by the Bloch ball, the shrinking factor of a given polytope is simply the radius of the largest sphere that fits inside it.
This radius is simply the minimum distance between facets of this polyhedron and the origin of the Bloch ball.
In higher dimensions the shrinking factor of a polytope approximating the set of quantum states can also be easily computed if we have its facet description.
Indeed, consider a polytope with hermitian facets $\Tr(F_i\rho)\leq b_i$ indexed by $i$.
Denoting $\Sp(F)$ the spectrum of $F$, the equality $\max\{\Tr(F_i\rho)~|~\rho\geq0,\Tr\rho=1\}=\max\Sp(F)$ is straightforward and allows us to derive the polytope's shrinking factor $r$, namely,
\begin{equation}
  r=\min_i\frac{db_i-\Tr(F_i)}{d\max\Sp(F_i)-\Tr(F_i)}.
  \label{eq:shr}
\end{equation}
Since our method needs the vertices of the polytope and \cref{eq:shr} its facets, it is necessary to interchange between the polytope's vertex and facet descriptions, which limits the reach of this technique.
This is nonetheless the main tool used in our work to compute shrinking factors.

Note that \cite[Lemma~9.2]{AS17} can provide bounds on the shrinking factor even when the facets of the polytope defined by the set of density matrices are not known, using instead a good approximation of the complex sphere, easier to characterize than the set of quantum states.
In practice, these bounds are not tight enough for our purposes, leading to structures with too many vertices to be useful in our scheme.

However, using this method we can give an estimate of the number of vertices required to reach a certain shrinking factor, which will turn useful in the complexity analysis below.
By using spherical coordinates in the projective real $(2d-1)$-dimensional sphere (identified with the projective complex $d$-dimensional sphere), we can subdivide all angles into $t$ pieces as in~\cite{hirsch16} before using the aforementioned lemma, yielding a polytope with
\begin{equation}\label{eq:angular_subdivision}
    n=\frac{(t-1)^{2d-1}-1}{t-2}\quad\text{pure states and shrinking factor}\quad r\geq2\cos^{4(d-1)}\left(\frac{\pi}{2t}\right)-1.
\end{equation}
For further reference, we note that, asymptotically, $n\sim t^{2(d-1)}$ so that the upper bound on $1-r$ obtained from \cref{eq:angular_subdivision} is equivalent to $\pi^2(d-1)n^{-1/(d-1)}$.

The following theorem expresses how the shrinking factor is associated to approximations of the steering depolarizing robustness.
\begin{theorem}\label{thm:shrinking_sandwich}
  Given a set of states $\{\rho_\lambda\}$ acting on a $d$-dimensional Hilbert space, which includes $\frac{\1}{d}$ in the interior of its convex hull and with shrinking factor $r$, it holds that, for any $d$-dimensional assemblage $\{\sigma_{a|x}\}$,
  \begin{equation}\label{eq:shrinking_sandwich}
    \tilde{\eta}(\{\sigma_{a|x}\}, \{\rho_\lambda\}) \leq \eta^*(\{\sigma_{a|x}\}) \leq \frac{\tilde{\eta}(\{\sigma_{a|x}\}, \{\rho_\lambda\})}{r}.
  \end{equation}
\end{theorem}
\begin{proof}
  Notice that by the definition of the depolarizing map, given in \cref{eq:depolarizing_map}, for any $\eta, \xi \in \mathbb{R}$, it holds that $\Lambda_\eta \circ \Lambda_\xi = \Lambda_{\eta \xi}$.
  Thus, since for any $d$-dimensional quantum state $\rho$ one can write $\Lambda_r(\rho) = \sum_\lambda q_\lambda \rho_\lambda$, with $q_\lambda \geq 0$ and $\sum_\lambda q_\lambda = 1$, it follows, by applying the map $\Lambda_{1/r}$ on both sides of this equality, that $\rho = \sum_\lambda q_\lambda \Lambda_{1/r}(\rho_\lambda)$.
  This means that any $d$-dimensional quantum state lies in the convex hull of the set $\{\Lambda_{1/r}(\rho_\lambda)\}$, which implies $\tilde{\eta}(\{\sigma_{a|x}\}, \{\Lambda_{1/r}(\rho_\lambda\})) \geq \eta^*(\{\sigma_{a|x}\})$.
  Additionally, note that for any probability distributions $p(\lambda)$ and $\{p(a|x, \lambda)\}$ the equations
  \begin{equation}
    \Lambda_\eta(\sigma_{a|x}) = \sum_\lambda p(\lambda)p(a|x, \lambda)\rho_\lambda
  \end{equation}
  and
  \begin{equation}
    \Lambda_{\eta/r}(\sigma_{a|x}) = \sum_\lambda p(\lambda)p(a|x, \lambda)\Lambda_{1/r}(\rho_\lambda)
  \end{equation}
  are equivalent, which leads to
  \begin{equation}
      \tilde{\eta}(\{\sigma_{a|x}\}, \{\Lambda_{1/r}(\rho_\lambda)\}) = \frac{\tilde{\eta}(\{\sigma_{a|x}\}, \{\rho_\lambda\})}{r}.
  \end{equation}
  This proves the second inequality in \cref{eq:shrinking_sandwich}, while the first one follows straightforwardly from the definition \eqref{eq:def_approx_robustness}, as discussed in the previous subsection.
\end{proof}

In other words, the theorem above states that a set of states $\{\rho_\lambda\}$, when used in the linear program \eqref{eq:main_LP}, gives both a lower and an upper bound on the estimated steering robustness.
Moreover, the ratio between these bounds, and therefore the tolerance of the bounds computed by the LP, is associated with the shrinking factor of the set $\{\rho_\lambda\}$.
Hence, the closer the shrinking factor of this set is to one, the tighter the bounds obtained for the steering robustness of \textit{any} assemblage.

\subsection{Complexity analysis}
The relation between a parameter of the polytope $\{\rho_\lambda\}$ and the tolerance of the bounds provided by the linear program (given by \cref{thm:shrinking_sandwich}) allows us to formalize the computational complexity of the method we propose.
For simplicity, let us fix the number of outcomes $k$ and dimension $d$ and let us keep using $m$ to denote the number of measurements under study.
Throughout, we work in the \emph{unit-cost Random Access Machine (RAM) model} \citep{renegar2001mathematical,boyd2004convex} in which every machine word can store a signed integer of $O(\log (1/\epsilon))$ bits and every arithmetic or logical operation on such words costs~$1$.

On the one hand, the standard SDP \eqref{eq:SDP_steering} has $k^m$ semidefinite variables of size $d\times d$.
Interior-point methods therefore run in
\begin{equation}
  \label{eq:complexity_SDP_generic}
  T_{\mathrm{SDP}}(m,\epsilon)=
  \poly(d,k)\,
  2^{\Theta(m)}\,
  \log(1/\epsilon)
\end{equation}
RAM steps, where the $\log(1/\epsilon)$ factor is the well-known dependence of interior-point solvers on the target additive accuracy~$\epsilon$.

On the other hand, after fixing a polytope with $n$ vertices, our linear program \eqref{eq:main_LP} contains $O(kmn)$ real (scalar) variables and constraints, so a generic interior-point LP solver takes
\begin{equation}
  T_{\mathrm{LP}}(m,n,\epsilon)= \poly(k,m,n)\, \log(1/\epsilon)
\end{equation}
RAM steps.
In all applications below we keep $k$ and $n$ \emph{polynomial} in~$m$, whence $T_{\mathrm{LP}}(m,n,\epsilon)=\poly(m)\,\log(1/\epsilon)$.
Note that since extremal POVMs have at most $d^2$ outcomes~\cite{chiribella06}, many applications restrict to this value of $k\leq d^2$.

The price we pay for this exponential speed-up is an $\epsilon$-dependent loss in precision that is controlled by the shrinking factor $r$ of the chosen polytope: from \cref{thm:shrinking_sandwich}, the additive error on the incompatibility (or steering) robustness never exceeds $r^{-1}-1$.
Using the construction of \cref{eq:angular_subdivision}, we can actually define a hierarchy of polytopes (which leads to a hierarchy of LPs) with an increasing number of vertices $n$, such that $r(n)^{-1}-1=O\!\bigl(n^{-1/(d-1)}\bigr)$,
so choosing $n=\Theta(1/\epsilon^{d-1})$ keeps the geometric error below~$\epsilon$ while preserving polynomial running time.

With this choice we can make the complexity bound of the LP hierarchy in $m$ and $\epsilon^{-1}$ explicit.
First, resolving the $\poly(k,m,n)$ factor, we get a cost of
\begin{equation}
  T_{\mathrm{LP}}(m,n,\epsilon)= O\!\bigl((kmn)^{3/2}\bigr)\,\log(1/\epsilon)
\end{equation}
RAM steps with an interior-point LP solver; see \cite{nesterov1994interior,Renegar1988,Wright1997} for details.
Setting $n=\bigl\lceil 1/\epsilon^{d-1}\bigr\rceil$, the choice that keeps the geometric error below~$\epsilon$ via the shrinking factor estimate $r(n)^{-1}-1=O\!(n^{-1/(d-1)})$, now gives the explicit bound
\begin{equation}\label{eq:complexity_LP_explicit}
  T_{\mathrm{LP}}(m,\epsilon)= O\!\bigl(m^{3/2}\,\epsilon^{-3(d-1)/2}\log(1/\epsilon)\bigr),
\end{equation}
with the constant hidden in the big-$O$ depending only on the fixed number of outcomes~$k$, the fixed dimension $d$, and the interior-point parameters.

Thus the overall running time is polynomial in~$m$, a significant improvement over the exponential $T_{\mathrm{SDP}}=\tilde{O}(2^{m})$ required by the standard SDP formulation.
However, now the running time is also polynomial in $\epsilon^{-1}$, which is exponentially worse than the logarithmic scaling of $T_{SDP}$ with $\epsilon^{-1}$.
Therefore, this formally states that the linear program \eqref{eq:main_LP} indeed trades off accuracy for efficiency.

In practice, this trade off is convenient for several instances, especially for large $m$.
While the standard SDP approach is limited to a considerably small number of measurements,
we will later show that the LP is able to provide tight bounds for many previously intractable instances --- e.g., $m\sim400$ qubit measurements --- in a matter of seconds using a standard laptop.
Before presenting practical examples of the method, however, we must first address the question of which polytopes to use in the linear program.

\subsection{How to choose a good finite approximation for the set of states}
\label{sec:polytopes}
In this subsection, we discuss the problem of which sets of states $\{\rho_\lambda\}$ are good choices to be used in the linear program \eqref{eq:main_LP}.
We focus specifically on sets of \textit{states}, i.e., $\rho_\lambda \geq 0$ and $\Tr(\rho_\lambda) = 1$, which in addition to giving lower bounds for the estimated steering robustness, can also be used to obtain upper bounds for it, as \cref{thm:shrinking_sandwich} shows.
Moreover, this theorem also shows that the ratio between these bounds is related to the shrinking factor of the state set.

Therefore, we must aim at constructing sets of states with high shrinking factors.
Unfortunately, constructing such sets in arbitrary dimensions is a complicated task, and there is not a known method that consistently does it.
Even so, let us discuss some ideas to obtain reasonable shrinking factors in various dimensions.

To begin with, in the simplest case of qubits, the fact that the state space can be represented by the three-dimensional Euclidean ball is remarkably helpful, since many problems related to `uniformly' distributing points on a sphere have already been thoroughly investigated in the literature (see, for example, Ref.~\cite{HardinSloane_spherical_codes}).
To the best of our knowledge, however, it seems like the problem of finding optimal configurations in terms of the shrinking factor has not been previously analyzed.
Fortunately, even though constructed with other aims (such as maximizing the volume or the minimal distance between the points), the polytopes available in the literature turn out to have excellent shrinking factors.

\begin{table}[H]
  \centering
  \begin{tabular}{cc}
    \toprule
    Number of vertices & Shrinking factor \\
    \midrule
    92   & 0.9716 \\
    162  & 0.9842 \\
    252  & 0.9899 \\
    432  & 0.9942 \\
    612  & 0.9959 \\
    1212 & 0.9979 \\
    5532 & 0.9995 \\
    8192 & 0.9997 \\
    \bottomrule
  \end{tabular}
  \caption{
    Properties of some polytopes obtained from the covering problem under icosahedral symmetry~\cite{HardinSloane_spherical_codes}.
  }
  \label{tab:sloane_shr}
\end{table}

\begin{figure}[ht]
  \centering
  \begin{subfigure}{0.45\textwidth}
    \includegraphics[width=\textwidth,trim={30cm 15cm 30cm 15cm},clip]{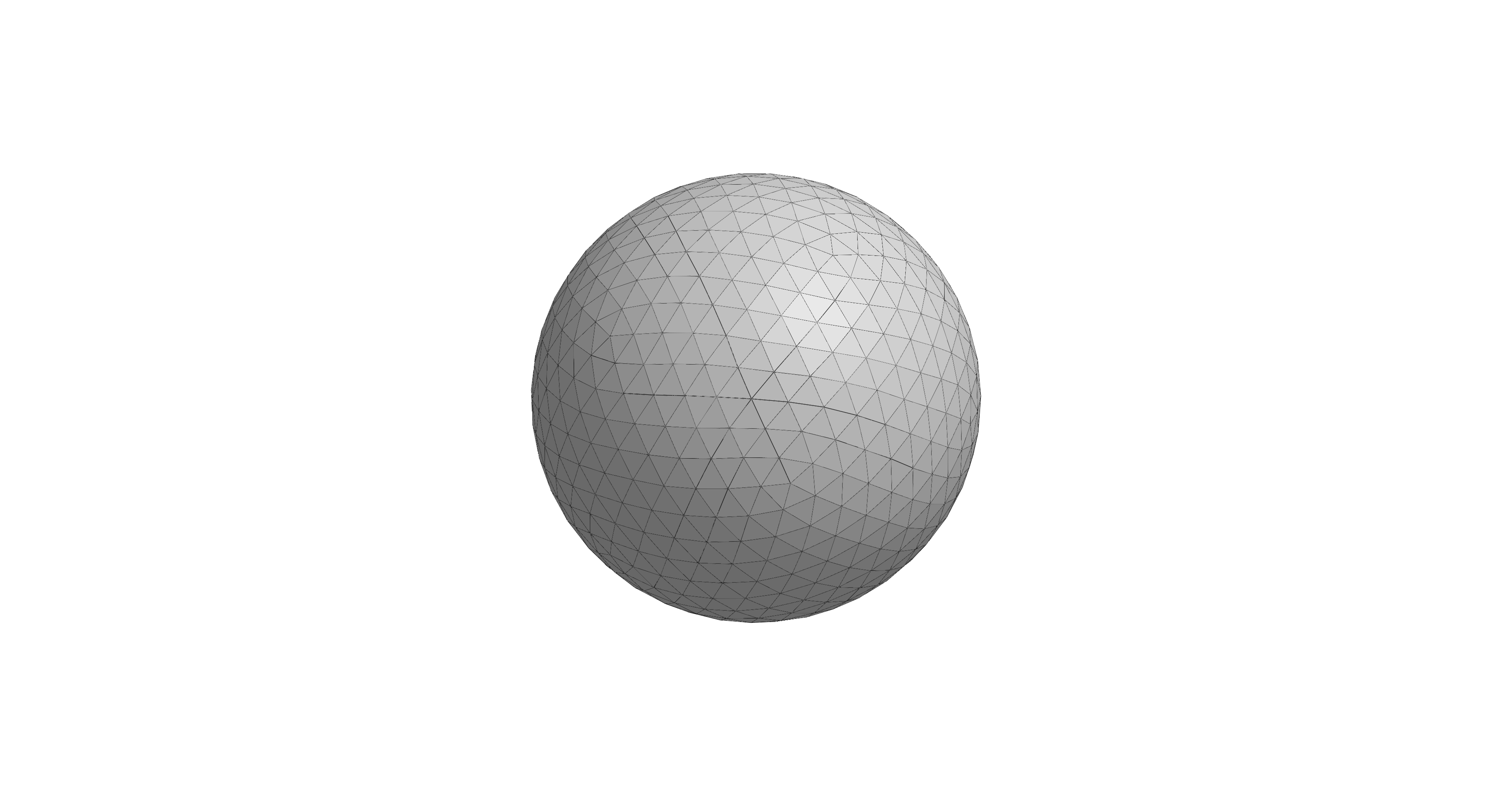}
    \caption{}
    \label{fig:polyhedron_sloane612}
  \end{subfigure}
  \hspace{1cm}
  \begin{subfigure}{0.45\textwidth}
    \includegraphics[width=\textwidth,trim={30cm 15cm 30cm 15cm},clip]{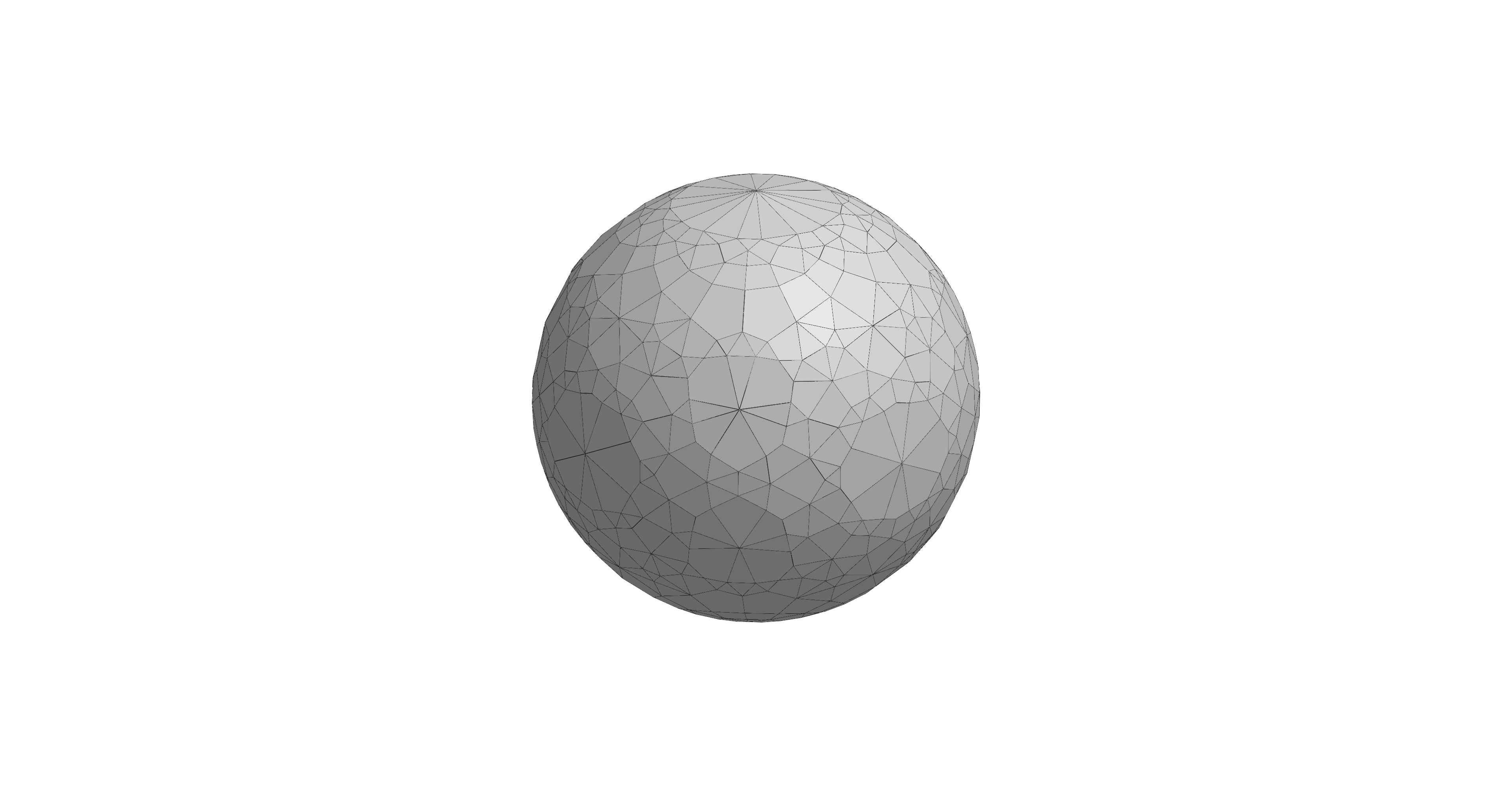}
    \caption{}
    \label{fig:polyhedron_rational646}
  \end{subfigure}
  \caption{
    Polytope approximations of the three-dimensional Euclidean ball.
    (a) Polytope with 612 vertices and shrinking factor 0.9959 from Ref.~\cite{HardinSloane_spherical_codes}.
    (b) Rational polytope $R_{2,27}$ with 646 vertices and shrinking factor 0.9800.
  }
\end{figure}

In particular, the best polytopes in terms of shrinking factor that came up in our search are arrangements that address a version of the so-called \textit{covering problem}, studied in Ref.~\cite{HardinSloane_spherical_codes}.
In \cref{tab:sloane_shr}, we exhibit the shrinking factors of some of these polytopes, and \cref{fig:polyhedron_sloane612} illustrates one of them.
These will be the main qubit polytopes we use to analyze the examples in \cref{sec:examplesJM,sec:examples_steering}.

In higher dimensions, however, the problem is not so simple.
As we have mentioned, in these cases the problem of finding sets of states with a good shrinking factor becomes increasingly more difficult~\cite{Fawzi18}.
The state space itself becomes a set which is geometrically way less understood.
Additionally, computing the shrinking factors of large polytopes in higher dimensions also becomes a complicated task, since our main technique to do so involves facet enumeration.

To the best of our knowledge, the best readily usable software to perform facet enumeration is PANDA~\cite{LR15}.
It turns out that PANDA can only handle rational polytopes, but can exploit symmetry to accelerate the computations.
Then, to habilitate the use of PANDA and the computation of shrinking factors of larger polytopes, it is convenient to consider sets of \textit{rational} states.

Consider the set $R_{d,q}$ of $d$-dimensional rational pure states with denominator at most $q$, e.g.,
\begin{equation}
  \frac15\begin{pmatrix}
    4 & 2 \\
    2 & 1
  \end{pmatrix}\in R_{2,5}
  \quad\text{or}\quad
  \frac19\begin{pmatrix*}[r]
    4 & -2 & 4 \\
    -2 & 1 & -2 \\
    4 & -2 & 4 \\
  \end{pmatrix*}\in R_{3,9}.
\end{equation}
In \cref{tab:shr_rational}, we compute the shrinking factors of some of these sets in dimensions up to five, and \cref{fig:polyhedron_rational646} illustrates a polytope formed by one of these sets.

\begin{table}[ht]
  \centering
  \begin{tabular}{ccccccccc}
    \toprule
    \multirow{2}{*}{$q$} & \multicolumn{2}{c}{$d=2$} & \multicolumn{2}{c}{$d=3$} & \multicolumn{2}{c}{$d=4$} & \multicolumn{2}{c}{$d=5$} \\
    \cmidrule(lr){2-3} \cmidrule(lr){4-5} \cmidrule(lr){6-7} \cmidrule(lr){8-9}
    & Size & Shrinking & Size & Shrinking & Size & Shrinking & Size & Shrinking \\
    \midrule
    2  & 6    & 0.5774 & 15    & 0.3504                    & 28  & 0.2330                    & 45  & 0.1667                    \\
    3  & 14   & 0.7071 & 55    & 0.6168                    & 140 & 0.4625                    & 285 & \hphantom{$^*$}0.3533$^*$ \\
    4  & 14   & 0.7071 & 103   & 0.6168                    & 396 & \hphantom{$^*$}0.5469$^*$ &     &                           \\
    5  & 22   & 0.8165 & 175   & 0.7348                    &     &                           &     &                           \\
    6  & 38   & 0.8944 & 271   & 0.8051                    &     &                           &     &                           \\
    7  & 54   & 0.8944 & 511   & 0.8371                    &     &                           &     &                           \\
    8  & 54   & 0.8944 & 703   & 0.8437                    &     &                           &     &                           \\
    9  & 78   & 0.8944 & 919   & 0.8496                    &     &                           &     &                           \\
    10 & 94   & 0.9370 & 1207  & 0.8872                    &     &                           &     &                           \\
    11 & 118  & 0.9428 & 1807  & 0.9029                    &     &                           &     &                           \\
    12 & 118  & 0.9428 & 2191  & 0.9067                    &     &                           &     &                           \\
    13 & 142  & 0.9428 & 2695  & 0.9146                    &     &                           &     &                           \\
    14 & 174  & 0.9487 & 3271  & 0.9196                    &     &                           &     &                           \\
    15 & 206  & 0.9487 & 4231  & 0.9250                    &     &                           &     &                           \\
    16 & 206  & 0.9487 & 4999  & 0.9250                    &     &                           &     &                           \\
    17 & 238  & 0.9623 & 5863  & 0.9338                    &     &                           &     &                           \\
    18 & 286  & 0.9623 & 6727  & 0.9369                    &     &                           &     &                           \\
    19 & 326  & 0.9701 & 8527  & \hphantom{$^*$}0.9473$^*$ &     &                           &     &                           \\
    27 & 646  & 0.9800 &       &                           &     &                           &     &                           \\
    82 & 5632 & 0.9931 &       &                           &     &                           &     &                           \\
    99 & 8206 & 0.9939 &       &                           &     &                           &     &                           \\
    \bottomrule
  \end{tabular}
  \caption{
    Cardinality and shrinking factor of the set $R_{d,q}$ of rational $d$-dimensional pure states with denominator at most $q$.
    All values are analytical since they come from \cref{eq:shr} for facets with integer coefficients.
    In dimension $d=2$, we include some higher denominators for comparison with the sets used in \cref{subsec:examplesJM_qubits} (see \cref{tab:sloane_shr}).
    Stars indicate polytopes for which PANDA did not terminate, so that the value is only an upper bound.
  }
  \label{tab:shr_rational}
\end{table}

For the case of qubits, the shrinking factors of such rational polytopes are, as expected, considerably lower than the ones from Ref.~\cite{HardinSloane_spherical_codes}.
However, for qutrit polytopes with thousands of vertices, these rational polytopes are the only ones for which we could enumerate all facets, and thus the best be could find.
In \cref{app:shr_sym_qutrit}, we describe a construction of more symmetric (irrational) qutrit polytopes for which we could enumerate facets and observe better shrinking factors up to a few hundreds of vertices.

The limitations in the techniques here discussed are essentially computational.
If one wishes to go beyond them, analytical constructions most likely have to be employed.
Besides avoiding computational limits, handling it with pen an paper may motivate further studies on other symmetric structures.
For example, in \cref{app:mubs} we present an analytical construction based on complete sets of mutually unbiased bases (MUBs) valid in any dimension where they exist.
In the examples we study in the following sections, however, this construction is not sufficient to provide relevant bound using \cref{thm:shrinking_sandwich}.

Another alternative to avoid the computational difficulties related to shrinking factors is to search for sets of states $\{\rho_\lambda\}$ which provide good bounds \textit{only} for a specific assemblage under study.
That is, \cref{thm:shrinking_sandwich} characterizes the set of states which provide accurate bounds for arbitrary assemblages.
However, when a specific assemblage is considered, one may search for sets granting tight bounds just for it.
In fact, it follows from Caratheodory's theorem that for any fixed assemblage $\{\sigma_{a \vert x}\}$, there exists a small set of states $\{\rho_\lambda\}$ for which $\eta^*(\{\sigma_{a \vert x}\}) = \tilde{\eta}(\{\sigma_{a \vert x}\}, \{\rho_\lambda\}$) \cite{Skrzypczyk2020}.
The practical difficulty, however, lies in finding such states $\{\rho_\lambda\}$ for any given assemblage.

We have explored this possibility using the alternating optimization procedure proposed in Ref.~\cite{ohst2024adaptive}.
Unfortunately, the bounds obtained in such a way are hardly ever better than the ones obtained from the polytopes with the best shrinking factors we have.
Nonetheless, we believe further research is possible in this direction, especially targeted at finding the best state sets for any given assemblage.

\section{Applications to measurement incompatibility}\label{sec:examplesJM}
Let us now apply the linear program \eqref{eq:main_LP} to concrete examples.
In this section, in particular, we study the incompatibility of some sets of qubit and qutrit measurements, while examples related to steering are left to \cref{sec:examples_steering}.
For the case of qubits, we compute the incompatibility robustness of much larger sets than those treatable with the standard SDP approach \eqref{eq:main_SDP} with a considerably high precision.
This is mainly due to the existence in the literature of good polytope approximations of the Bloch sphere, as we discussed in \cref{sec:polytopes}.
In the case of qutrits, by using some of the polytopes we constructed, we obtain nontrivial bounds on incompatibility robustness of large sets of qutrit measurements, but which are not as tight as in the qubit case.
The Julia code we used for all the examples discussed below is available at \cite{gitcode}.

\subsection{Qubits}\label{subsec:examplesJM_qubits}
\subsubsection{Fibonacci projective measurements}
For every $m \in \mathbb{N}$, consider the set of measurements $\{M_{a|x}\}_{x = 1}^m$ where the projector $M_{1|x}$ is represented in the Bloch sphere by the vector $r_x \cos((x-1)\phi)\hat{i} + r_x \sin((x-1)\phi)\hat{j} + [1 - (x-1)/(m-1)]\hat{k}$, with $r_x = \sqrt{1 - [1 - (x-1)/(m-1)]^2}$ and $\phi = \pi(\sqrt{5} - 1)$.
This distribution of measurements in the Bloch sphere is inspired by the Fibonacci lattice \cite{gonzalez2010fibonacci}, and, for each $m$, these measurements are close to the most incompatible ones \cite{bavaresco17}.

In \cref{fig:qubit_proj,tab:hfib_incomp_rob}, we compute the incompatibility robustness of sets of measurements generated in this way for several values of $m$, using both the techniques developed in this work and, whenever possible, the SDP \eqref{eq:main_SDP}.
It is known that, for any $m$, the incompatibility robustness of such measurements is lower bounded by $1/2$ \cite{werner89} (see also \cite{zhang2024, renner2024} for recent developments), and we can interestingly observe the convergence to such a value when the number of measurements $m$ grows.

We also exhibit the corresponding execution times, estimated using a 8-core Intel$^\circledR$ Core$^{\text{TM}}$ i5-10400H with 64GB of RAM, and using Mosek~\cite{MOSEK} as the solver for both the linear and semidefinite programs.
As the set of states $\{\rho_\lambda\}$, required to run the linear program \eqref{eq:main_LP}, we use three polytopes from \cref{tab:sloane_shr}.
To each of them, we also added the vertices $\pm \hat{z}$, which increases the performance of the method, as we have empirically observed.
Another empirical observation is the fact that the execution time of the LP in general seems to depend on the alignment between the considered set of measurements and the chosen polytope.
We believe this is the reason why in our examples the execution time does not monotonically increase when more Fibonacci measurements are considered.

\subsubsection{Planar projective measurements}
Now, consider an arbitrary set of qubit projective measurements $\{M_{a\vert x}\}_{x = 1}^m$ defined by coplanar Bloch vectors $\vec{n}_x$.
In Ref.~\cite{andrejic2020}, the authors prove that
\begin{equation}\label{eq:coplanar_lower_bound}
    \eta^*(\{M_{a\vert x}\}) \geq \frac{1}{\sum_{x = 2}^{m}\sin(\frac{\alpha_x - \alpha_{x-1}}{2}) + \cos\frac{\alpha_m}{2}},
\end{equation}
where $\alpha_x$ is the angle between the vectors $\vec{n}_x$ and $\vec{n}_1$, ordered from smallest to largest \cite[Theorem 7]{andrejic2020}.
Given that equality in the above expression holds for $m = 2$, for $m = 3$, and for symmetric measurements, the authors conjecture that it must also hold in general \cite[Conjecture 2]{andrejic2020}.

To test this conjecture, we sampled sets of up to 500 planar measurements, and we computed their incompatibility robustness using the linear program \eqref{eq:main_LP}.
We note that since the measurements are coplanar, the elements of the parent measurement must also lie in the same plane.
Therefore, to use our linear program, we only need a polytope approximation of a circle, and not of the whole Bloch sphere.
This significantly improves the performance of the linear program, which in this case is able to provide lower and upper bounds matching in the fourth decimal digit using a polytope of only $\sim400$ states.

In total, we sampled $10^3$ measurement sets, and for all of them the bounds provided by the linear program are in agreement with \cite[Conjecture 2]{andrejic2020}, which significantly strengthens it.
Our methods therefore give strong numerical evidence that the inequality \eqref{eq:coplanar_lower_bound}, proven in \cite[Theorem 7]{andrejic2020}, is an equality, thus fully characterizing the incompatibility of coplanar projective qubit measurements.

\begin{figure}[H]
  \centering
  \begin{subfigure}{0.8\textwidth}
    \centering
    \includegraphics[width=\linewidth]{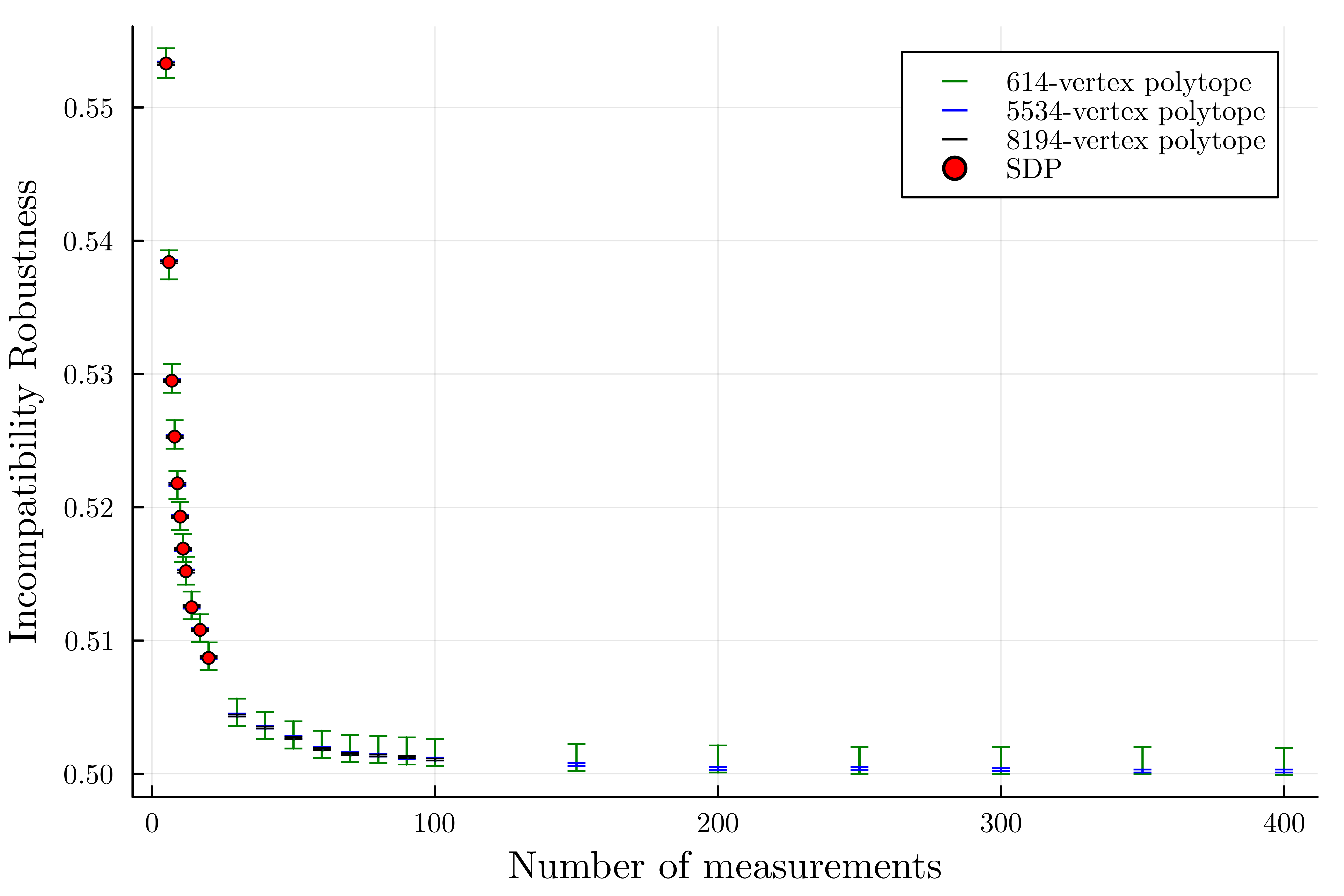}
  \end{subfigure}
  \begin{subfigure}{0.8\textwidth}
    \centering
    \includegraphics[width=\linewidth]{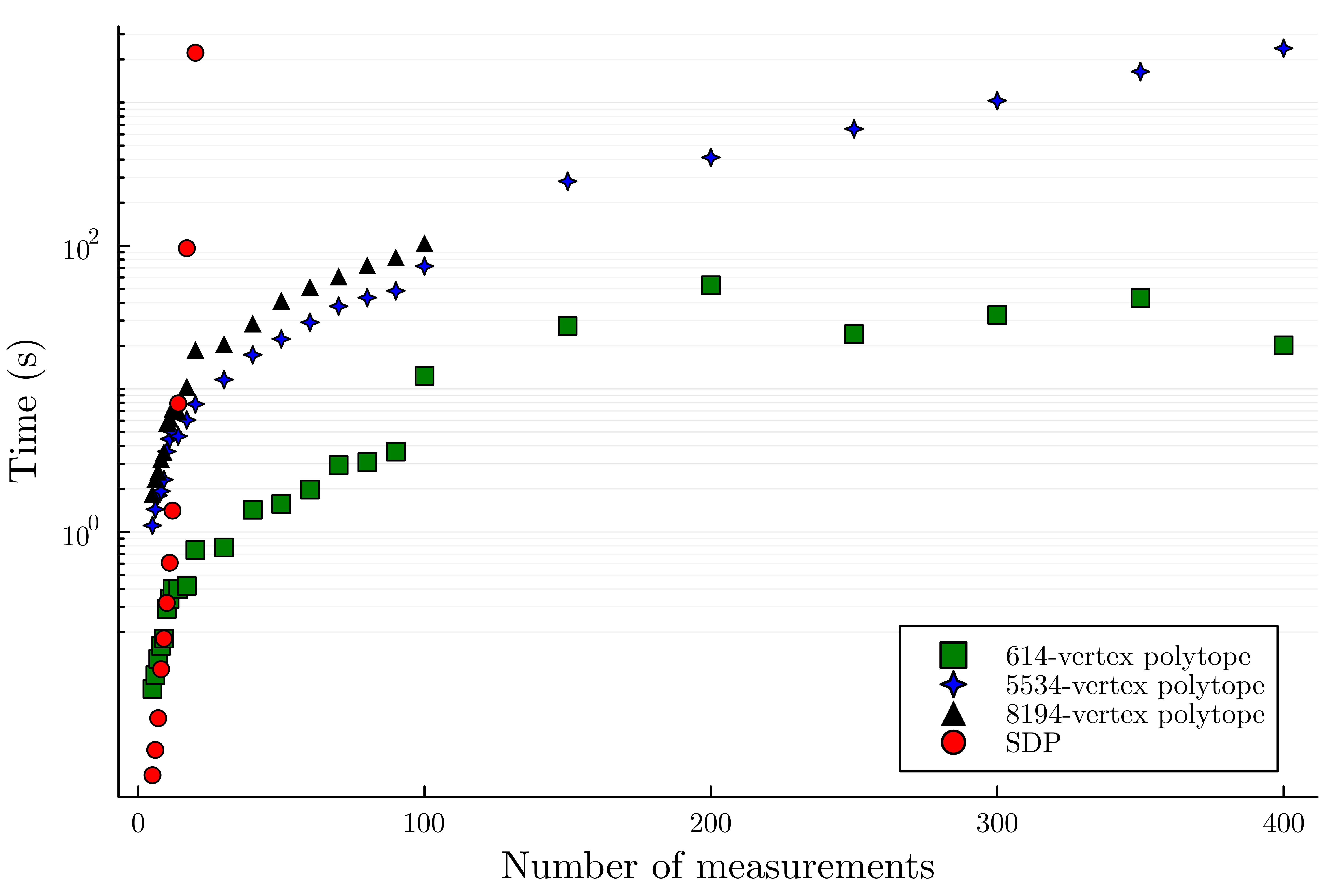}
  \end{subfigure}
  \caption{
    Comparison between results for the incompatibility robustness of qubit projective measurements obtained with the standard SDP \eqref{eq:main_SDP} and with the linear program \eqref{eq:main_LP}.
    The distributions of the measurements in the Bloch sphere was inspired by the Fibonacci lattice.
    Raw data can be found in \cref{tab:hfib_incomp_rob}.
  }
  \label{fig:qubit_proj}
\end{figure}

\begin{table}[H]
  \centering
  \scalebox{0.82}{
    \begin{tabular}{@{}ccccccccc@{}}
      \toprule
      & \multicolumn{2}{c}{SDP} & \multicolumn{2}{c}{614-vertex polytope} & \multicolumn{2}{c}{5534-vertex polytope} & \multicolumn{2}{c}{8194-vertex polytope} \\
      \cmidrule(lr){2-3} \cmidrule(lr){4-5} \cmidrule(lr){6-7} \cmidrule(lr){8-9}
      $m$ & $\eta$ & Runtime & $\eta$ & Runtime & $\eta$ & Runtime & $\eta$ & Runtime \\
      \midrule
      10  & 0.5193 & 0.32\,s       & [0.5183, 0.5204] & 0.29\,s  & [0.5192, 0.5194] & 3.65\,s       & [0.5192, 0.5194] & 5.76\,s      \\
      15  & 0.5118 & 15.50\,s      & [0.5109, 0.5130] & 0.39\,s  & [0.5117, 0.5119] & 5.41\,s       & [0.5117, 0.5119] & 7.41\,s      \\
      20  & 0.5087 & 37\,min 10\,s & [0.5078, 0.5098] & 0.75\,s  & [0.5086, 0.5089] & 7.83\,s       & [0.5087, 0.5088] & 12.74\,s     \\
      50  & ---    & ---           & [0.5019, 0.5039] & 1.57\,s  & [0.5026, 0.5028] & 22.32\,s      & [0.5026, 0.5027] & 41.34\,s     \\
      100 & ---    & ---           & [0.5006, 0.5026] & 12.37\,s & [0.5010, 0.5012] & 1\,min 12\,s  & [0.5010, 0.5011] & 1\,min 44\,s \\
      200 & ---    & ---           & [0.5001, 0.5021] & 53.00\,s & [0.5003, 0.5006] & 6\,min 54\,s  & ---              & ---          \\
      300 & ---    & ---           & [0.5000, 0.5020] & 32.86\,s & [0.5002, 0.5004] & 17\,min 11\,s & ---              & ---          \\
      400 & ---    & ---           & [0.4999, 0.5020] & 20.15\,s & [0.5001, 0.5004] & 39\,min 57\,s & ---              & ---          \\
      \bottomrule
    \end{tabular}
  }
  \caption{
    Incompatibility depolarizing robustness of qubit projective measurements distributed in the Bloch sphere in a way inspired by the Fibonacci lattice.
    For $m=21$, the SDP approach was not able to find a solution.
  }
  \label{tab:hfib_incomp_rob}
\end{table}

\newpage

\subsubsection{Random non-projective POVMs}\label{sec:nonproj_qubit}
As we have previously mentioned, the method we propose is not limited to projective measurements, it can also be applied to generic POVMs.
In \cref{fig:povm-qubits}, we exhibit the incompatibility robustness of random sets of extremal qubit POVMs, generated with the function \code{random_povm} from the package \texttt{Ket.jl}~\cite{Ket.jl}, which is based on Ref.~\cite{heinossari2020random}, and the algorithm proposed in Ref.~\cite{dariano05} to decompose any given POVM into extremal ones.
We compare their incompatibility with sets of random projective qubit measurements, generated in a similar manner, and with the Fibonacci measurements, studied in \cref{subsec:examplesJM_qubits}.
For each value of the number of measurements, we sampled five sets of random POVMs and five sets of projective measurements, and the values exhibited in \cref{fig:povm-qubits} correspond to the average incompatibility robustness we obtained.

Our results suggest that projective measurements are typically much more incompatible than non-projective POVMs.
This is in alignment with the evidences shown in Ref.~\cite{bavaresco17} that qubit projective measurements are more incompatible than non-projective POVMs, and with the fact that the incompatibility robustness of the set of all qubit projective measurements is the same as the robustness of the set of all qubit general POVMs, shown in Refs.~\cite{renner2024, zhang2024}.

\begin{figure}[H]
  \centering
  \includegraphics[width=0.78\linewidth]{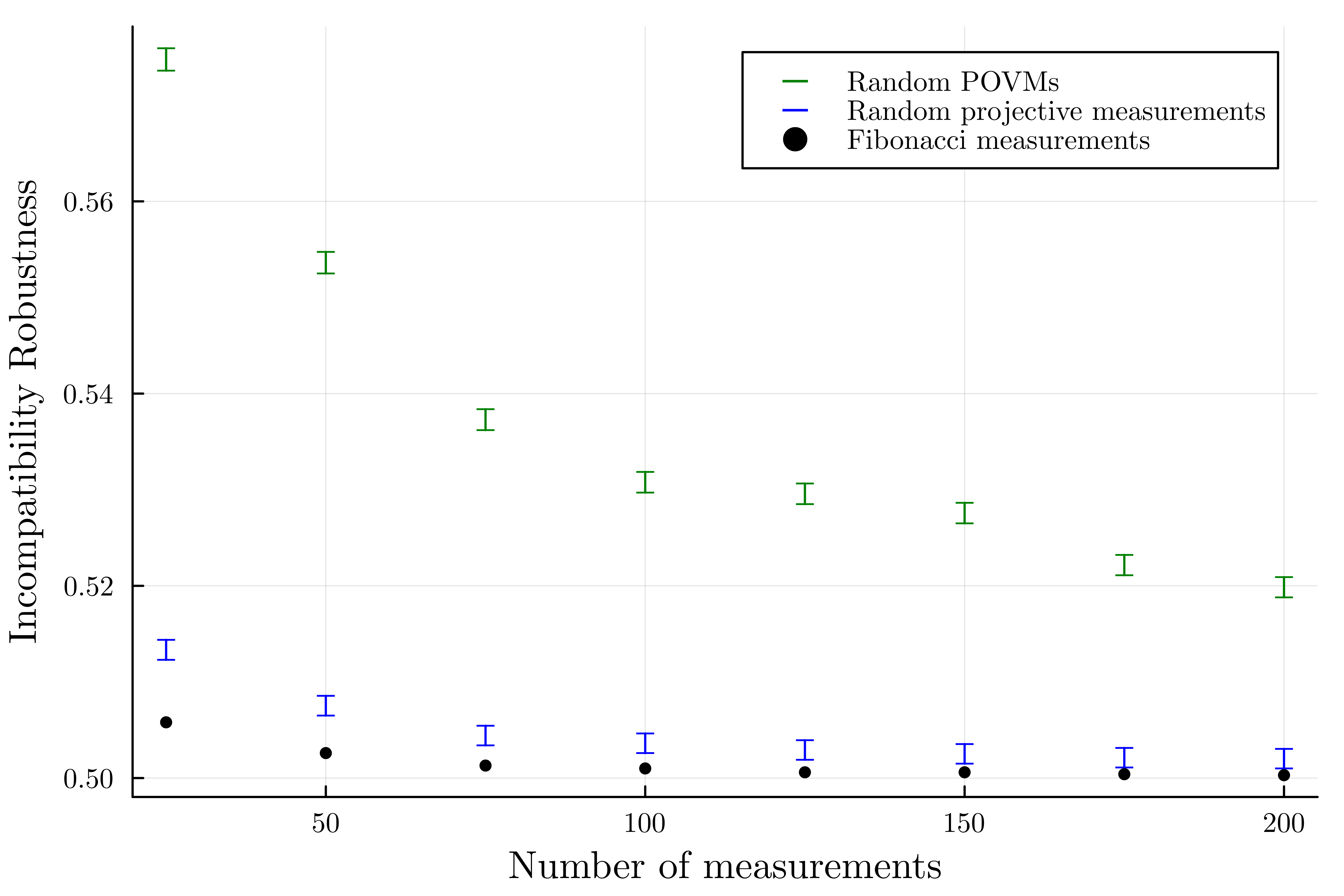}
  \caption{
    Incompatibility robustness of random extremal qubit POVMs, compared with random projective measurements and the Fibonacci measurements, studied in \cref{subsec:examplesJM_qubits}.
  }
  \label{fig:povm-qubits}
\end{figure}

\subsection{Qutrits}\label{sec:examples_JM_qutrits}
To illustrate our method for qutrits, we consider sets of projective measurements once more inspired by the Fibonacci lattice.
For every $m \in \mathbb{N}$, consider the set of trichotomic measurements (i.e., measurements with three outcomes) $\{M_{a \vert x}\}_{x = 1}^m$ where $M_{1 \vert x}$ is the projector onto the state $\ket{\psi_1}_x = \sqrt{1 - z_x^2}\cos(\theta_x)\ket{0} + \sqrt{1 - z_x^2}\sin(\theta_x)\ket{1} + z_xe^{i\phi_x}\ket{2}$, with $z_x = 1 - (x-1)/(m-1)$, $\theta_x = (x-1)\pi(\sqrt{5} - 1)$ and $\phi_x = \pi\frac{x-1}{m-1}$, $M_{2 \vert x}$ is the projector onto $\ket{\psi_2}_x = \sin(\theta_x)\ket{0} - \cos(\theta_x)\ket{1}$, and $M_{3 \vert x} = \1 - M_{1 \vert x} - M_{2 \vert x}$.
In \cref{fig:qutrit_measurements,tab:qutrit_fib_incomp_rob}, we exhibit the incompatibility robustness of these measurements for various values of $m$, computed using both the techniques we propose and the standard SDP approach, along with the respective execution times.
Notice that the bounds obtained with our techniques in this case are much looser than the ones obtained for qubit measurements (see \cref{subsec:examplesJM_qubits}).

\begin{figure}[H]
  \centering
  \begin{subfigure}{0.78\textwidth}
    \centering
    \includegraphics[width=\linewidth]{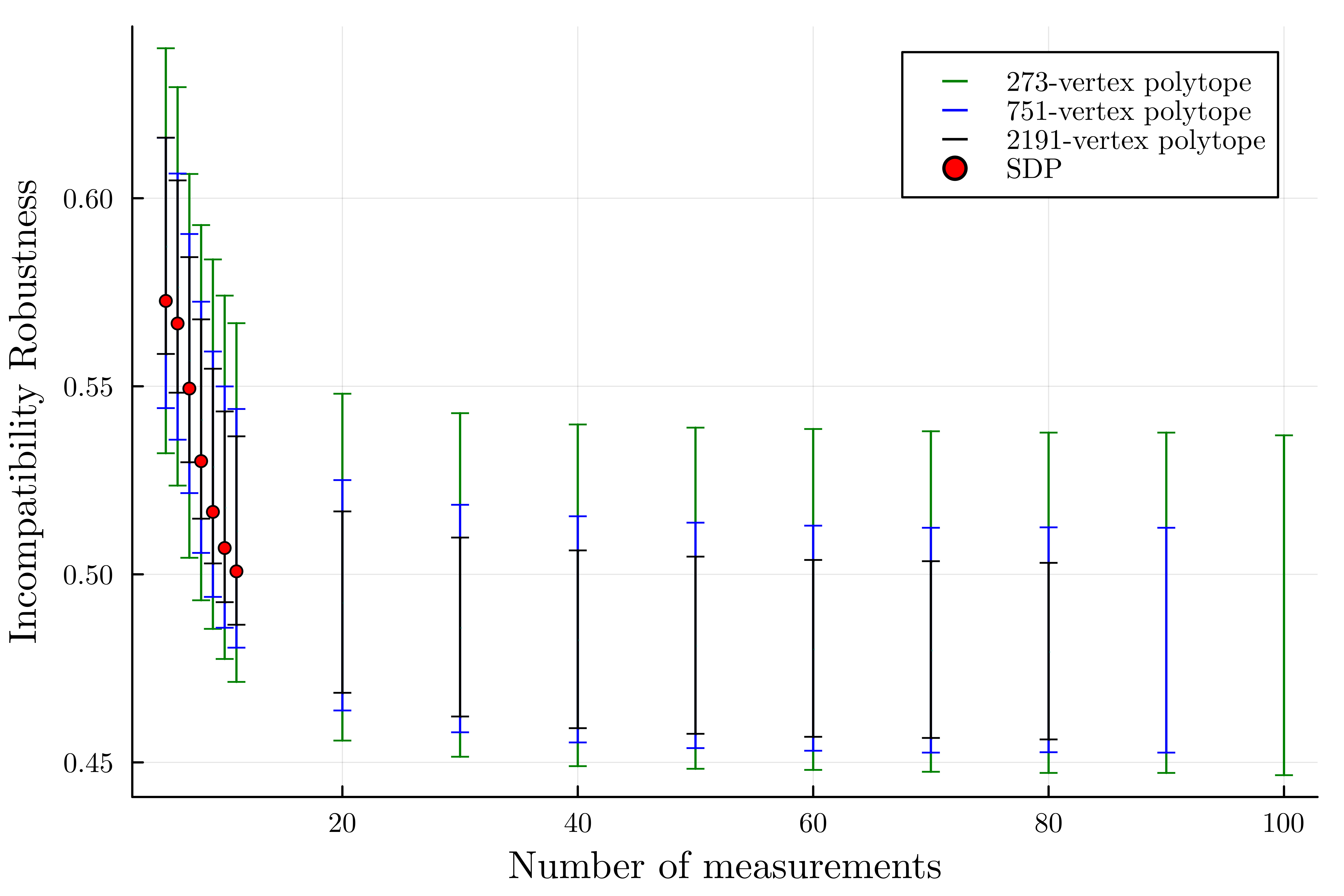}
    \label{fig:qutrit_visibilities}
  \end{subfigure}
  \vspace{-10pt}
  \begin{subfigure}{0.78\textwidth}
    \centering
    \includegraphics[width=\linewidth]{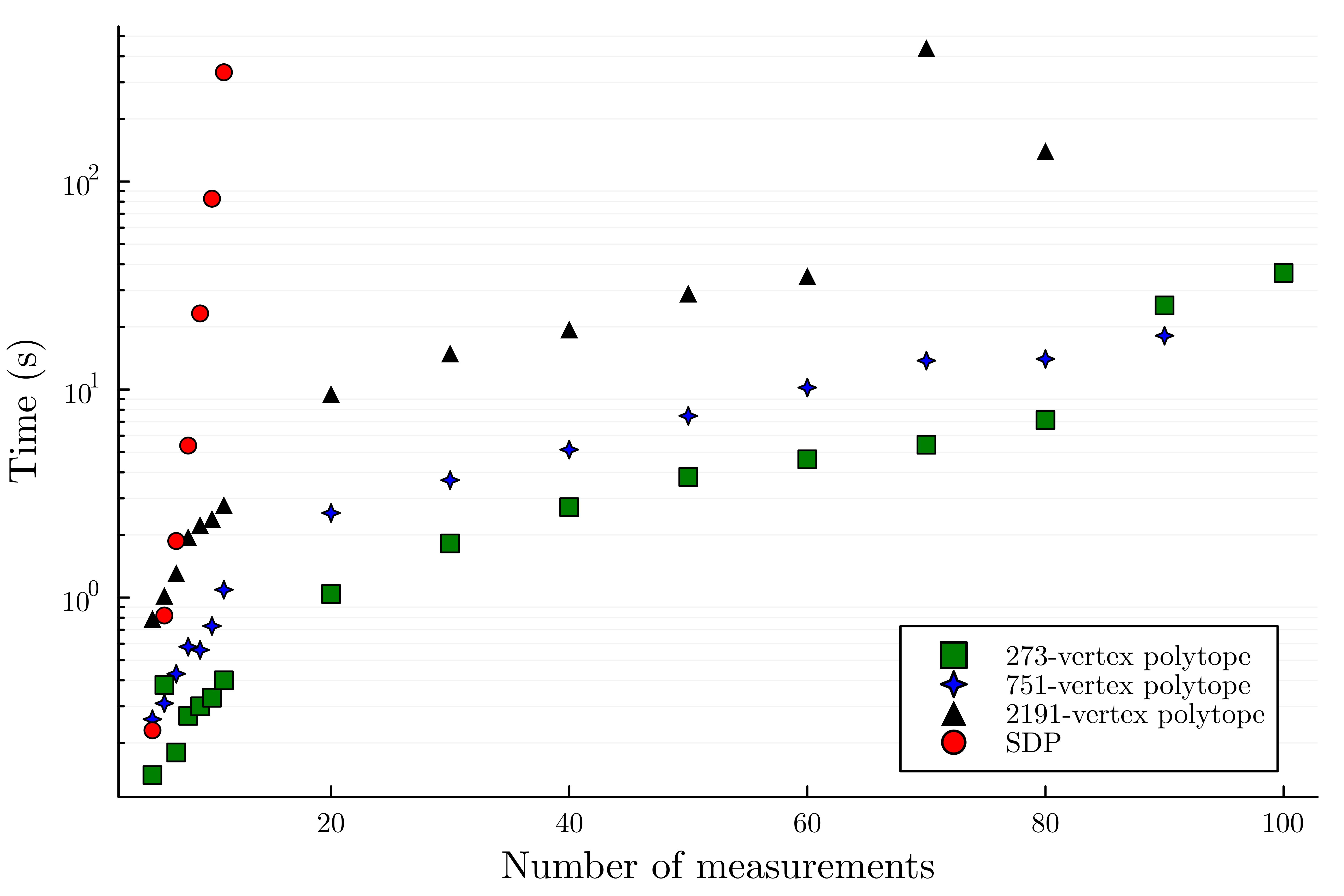}
    \label{fig:qutrit_runtimes}
  \end{subfigure}
  \vspace{-2pt}
  \caption{
    Comparison between results for the incompatibility robustness of qutrit projective measurements obtained with the standard SDP \eqref{eq:main_SDP} and with the linear program \eqref{eq:main_LP}.
    The measurements under study are inspired by the Fibonacci lattice.
    Raw data can be found in \cref{tab:qutrit_fib_incomp_rob}.
  }
  \label{fig:qutrit_measurements}
\end{figure}

\begin{table}[H]
  \centering
  \scalebox{0.84}{
    \begin{tabular}{@{}ccccccccc@{}}
      \toprule
      & \multicolumn{2}{c}{SDP} & \multicolumn{2}{c}{273-vertex polytope} & \multicolumn{2}{c}{751-vertex polytope} & \multicolumn{2}{c@{}}{2191-vertex polytope} \\
      \cmidrule(lr){2-3} \cmidrule(lr){4-5} \cmidrule(lr){6-7} \cmidrule(l){8-9}
      $m$ & $\eta$ & Runtime & $\eta$ & Runtime & $\eta$ & Runtime & $\eta$ & Runtime \\
      \midrule
      9   & 0.5166 & 23.22\,s      & [0.4855, 0.5837] & 0.30\,s  & [0.4940, 0.5592] & 0.56\,s      & [0.5029, 0.5547] & 2.23\,s      \\
      10  & 0.5070 & 1\,min 22\,s  & [0.4775, 0.5741] & 0.33\,s  & [0.4858, 0.5500] & 0.73\,s      & [0.4926, 0.5433] & 2.39\,s      \\
      11  & 0.5008 & 5\,min 35\,s  & [0.4714, 0.5668] & 0.40\,s  & [0.4802, 0.5436] & 1.09\,s      & [0.4866, 0.5367] & 2.78\,s      \\
      20  & ---    & ---           & [0.4558, 0.5480] & 1.04\,s  & [0.4638, 0.5251] & 2.55\,s      & [0.4685, 0.5167] & 9.50\,s      \\
      40  & ---    & ---           & [0.4490, 0.5399] & 2.72\,s  & [0.4553, 0.5154] & 5.14\,s      & [0.4591, 0.5064] & 19.44\,s     \\
      60  & ---    & ---           & [0.4480, 0.5386] & 4.62\,s  & [0.4531, 0.5129] & 10.22\,s     & [0.4568, 0.5038] & 35.07\,s     \\
      80  & ---    & ---           & [0.4472, 0.5377] & 7.13\,s  & [0.4527, 0.5125] & 14.03\,s     & [0.4561, 0.5030] & 2\,min 20\,s \\
      100 & ---    & ---           & [0.4466, 0.5370] & 36.39\,s & ---              & ---          & ---              & ---          \\
      \bottomrule
    \end{tabular}
  }
  \caption{
    Incompatibility depolarizing robustness of qutrit projective measurements inspired by the Fibonacci lattice.
    For $m=12$, the SDP approach was not able to find a solution.
  }
  \label{tab:qutrit_fib_incomp_rob}
\end{table}

\section{Applications to the steerability of quantum states}\label{sec:examples_steering}
In this section, we discuss how the linear program \eqref{eq:main_LP} can be used to decide whether a quantum state can or cannot lead to steerable assemblages.
Broadly speaking, the characterization of steering can be studied from a few different perspectives~\cite{Uola2020SteeringReview}.
In the first place, one can analyze whether correlations in the measurement statistics of the experiment can prove steerability (see, for example, Refs.~\cite{saunders10, cavalcanti09}).
Secondly, the starting point of the analysis can be a steering assemblage, which leads to the SDP characterization given by \cref{eq:SDP_steering}.
Furthermore, one may also aim at characterizing which quantum states can give rise to steering~\cite{quintino15, hirsch16, cavalcanti15, Nguyen2019LHSCriterion}.

The techniques proposed in this work can be immediately applied within the second approach mentioned above, since the linear program \eqref{eq:main_LP} can replace the SDP \eqref{eq:SDP_steering}.
These uses of the method, however, are very similar to the ones discussed in the last section.
Indeed, calculating the incompatibility robustness of a set of measurements is equivalent to calculating the steering robustness of the assemblage produced by applying those measurements in a maximally entangled state~\cite{quintino15, uola15}.
For this reason such examples will not be further discussed here.
Rather, in this section we study how the LP \eqref{eq:main_LP} is useful in the third of the aforementioned approaches to characterize steering.
That is, we analyze how the proposed LP can be used to determine whether a given quantum state can produce steerable correlations.

To set the nomenclature, we say that a bipartite quantum state $\rho_{AB}$ is \textit{steerable} if there exists a set of measurements $\{M_{a|x}\}$ such that the assemblage $\sigma_{a|x} = \Tr_A(M_{a|x} \otimes \1~\rho_{AB})$ does not admit an LHS model.
On the other hand, if for any set of measurements of a particular kind (e.g., projective, POVMs) the resulting assemblage admits an LHS model, the state is said to be \textit{unsteerable} with respect to the specified type of measurement.

Similarly to what we have done in previous sections, instead of simply attempting to decide whether a given quantum state is unsteerable or steerable, we aim at computing a kind of steering robustness at the level of states.
More precisely, given a bipartite state $\rho_{AB}$ of local dimension $d$, we consider the quantity
\begin{equation}\label{eq:def_state_steering_rob}
  \eta^*(\rho_{AB}) = \max \left\{\eta ~\big|~ \Lambda_\eta \otimes \1 (\rho_{AB}) = \eta \rho_{AB} + (1 - \eta) \frac{\1}{d} \otimes \rho_B \text{ is unsteerable}\right\},
\end{equation}
where $\rho_B = \Tr_A(\rho_{AB})$, and unsteerability is to be understood with respect to a certain type of measurements specified in each context.
Unlike the steering robustness of assemblages, the above quantity cannot be computed by an SDP.
In fact, the best available methods to study this problem are only able to compute lower and upper bounds on $\eta^*(\rho_{AB})$~\cite{hirsch15, cavalcanti15, Nguyen2019LHSCriterion}.

In the following, we show how the linear program \eqref{eq:main_LP} can be useful for this matter.
First, it can be straightforwardly used to compute upper bounds on $\eta^*(\rho_{AB})$ for states in any dimension, and with respect to arbitrary types of measurements.
Furthermore, we also show how it can be used to improve the method proposed in Refs.~\cite{hirsch15, cavalcanti15} to compute lower bounds on $\eta^*(\rho_{AB})$, also in arbitrary dimensions and for arbitrary types measurements.
To demonstrate how powerful these methods are, we compare their results for partially entangled two-qubit states with the bounds obtained from the method proposed in Ref.~\cite{Nguyen2019LHSCriterion}, which is currently the best available tool to study steerability of two-qubit states.

\subsection{Proving that a quantum state is steerable}
To begin with, in order to prove that a given bipartite state $\rho_{AB}$ is steerable, in principle one just has to find suitable measurements such that $\Tr_A(M_{a|x} \otimes \1~\rho_{AB})$ does not admit an LHS model.
Nevertheless, it is in general not clear how to find these measurements for any given quantum state.
In this sense, one sensible idea is to simply select a certain finite set of measurements, possibly at random, and test for the steerability of the resulting assemblage using the SDP \eqref{eq:main_SDP}.
However, since this SDP can only be used in practice for a small number of measurements, often the result is not accurate.

On the other hand, if the linear program \eqref{eq:main_LP} is used to test the steerability of the resulting assemblage, a much larger set of measurements can be chosen from the beginning.
Using a larger number of measurements, one is expected to better cover the set of measurements of interest, which subsequently leads to a more precise answer.

Therefore, to prove the steerability of a quantum state, or, more precisely, to compute an upper bound on $\eta^*(\rho_{AB})$, we propose the following method.
First, choose (possibly at random) a large set of measurements $\{M_{a|x}\}$.
Then, choose an outer approximation of the set of quantum states in the desired local dimension$\{\chi_\lambda\}$\footnote{In the sense that every quantum state in the corresponding dimension is contained in the convex hull of $\{\chi_\lambda\}$, see \cref{fig:out}.} and compute
\begin{subequations}\label{eq:LP_steerability}
  \begin{align}
    \tilde{\eta}(\rho_{AB}, \{M_{a|x}\},  \{\chi_\lambda\}) = \max \; &\eta \\
    \text{s.t.} \; &\Tr_A(M_{a|x} \otimes \1~\Lambda_\eta \otimes \1(\rho_{AB})) = \sum_{\lambda = 1}^n p(\lambda, a|x)\chi_\lambda, \quad \forall a,~x \label{eq:lp_steerability_lhs_constraint}\\
    &p(\lambda, a|x) \geq 0, \quad \forall \lambda,~a,~x\\
    &\sum_{a = 1}^k p(\lambda, a|x) = \sum_{a = 1}^k p(\lambda, a|x'), \quad \forall \lambda,~x,~x'.
  \end{align}
\end{subequations}
This guarantees that the state $\Lambda_\eta(\rho_{AB})$ is steerable for $\eta > \tilde{\eta}(\rho_{AB}, \{M_{a|x}\}, \{\chi_\lambda\})$.
That is, it holds that $\eta^*(\rho_{AB}) < \tilde{\eta}(\rho_{AB}, \{M_{a|x}\}, \{\chi_\lambda\})$.
The better the set $\{\chi_\lambda\}$ approximates the set of quantum states, and the more well-distributed the measurements $\{M_{a|x}\}$ among the set of measurements of interest, the better the computed upper bound is expected to be.

\subsection{Proving that a quantum state is unsteerable}\label{subsec:steerability}
Conversely, now consider the problem of proving that a quantum state is unsteerable with respect to \textit{all} measurements of a certain kind (e.g., projective measurements, dichotomic POVMs, etc.).
Perhaps the biggest difficulty in doing so is the fact that one has to deal with a continuous set of measurements, unlike what we have been considering so far.
Due to this difficulty, this problem has been firstly tackled for some symmetric classes of quantum states, for which explicit LHS models could be constructed from symmetry considerations~\cite{werner89, barrett02}.
However, later an algorithmic method was developed in order to prove the unsteerability of generic states~\cite{cavalcanti15, hirsch15}.

The linear program \eqref{eq:main_LP} is able to substantially improve this algorithm.
In order to describe how, let us briefly discuss the original algorithm's main idea.
For presentation purposes, we restrict the discussion to the case of two-qubit states under projective measurements, but it should be noted that similar ideas can be applied to states of any dimension and measurements of any kind.
One should begin by choosing a finite set of projective qubit measurements $\{M_{a|x}\}$ in the Bloch sphere with shrinking factor $\mu$.
Then, notice that the set of \textit{quasi}-measurements $\{\Lambda_{1/\mu}(M_{a \vert x})\}$ is an outer approximation of the Bloch sphere, in the sense that any other measurements therein belongs to its convex hull.
Therefore, if a two-qubit state $\rho_{AB}$ is such that the assemblage ${\Tr_A[\Lambda_{1/\mu}(M_{a|x}) \otimes \1~\rho_{AB}]}$ is unsteerable, it follows that the state $\rho_{AB}$ is unsteerable with respect to \textit{all} projective measurements.

The larger the shrinking factor of the set $\{M_{a|x}\}$, the more powerful the method \cite{hirsch15}.
Nevertheless, notice that a part of the algorithm consists in proving the unsteerability of an assemblage generated by such measurements, and in practice this was done using the SDP \eqref{eq:main_SDP}.
This poses a severe constraint on the number of measurements one could consider, which leads to fairly small shrinking factors.
In this context, instead of using the SDP \eqref{eq:main_SDP} to verify such unsteerability, we propose to use the linear program \eqref{eq:main_LP}.
This allows the method to be used with a much larger set of measurements, with significantly better shrinking factor.

In summary, the method we propose to prove the unsteerability of a quantum state, or, more precisely, to compute a lower bound on $\eta^*(\rho_{AB})$ of a given quantum state $\rho_{AB}$ goes as follows.
In the first place, choose a set of projective measurements $\{M_{a|x}\}$ on the Bloch sphere with shrinking factor $\mu$, and a set of states $\{\rho_\lambda\}$.
Then, compute the following linear program
\begin{subequations}\label{eq:LP_unsteerability}
  \begin{align}
    \tilde{\eta}(\rho_{AB}, \{M^{1/\mu}_{a \vert x}\}, \{\rho_\lambda\}) = \max \; &\eta \\
    \text{s.t.} \; &\Tr_A( M^{1/\mu}_{a|x} \otimes \1~\Lambda_\eta \otimes \1(\rho_{AB})) = \sum_{\lambda = 1}^n p(\lambda, a|x)\rho_\lambda, \quad \forall a,~x \label{eq:lp_unsteerability_lhs_constraint}\\
    &p(\lambda, a|x) \geq 0, \quad \forall \lambda,~a ,~x \\
    &\sum_{a = 1}^k p(\lambda, a|x) = \sum_{a = 1}^k p(\lambda, a|x'), \quad \forall \lambda,~x,~x',
  \end{align}
\end{subequations}
where $M^{1/\mu}_{a \vert x} = \Lambda_{1/\mu}(M_{a \vert x})$.
Then, it follows that the state $\Lambda_\eta(\rho_{AB})$ is unsteerable for $\eta \leq \tilde{\eta}(\rho, \{\Lambda_{1/\mu}(M_{a \vert x})\}, \{\rho_\lambda\})$.
Or, in other words, we have $\eta^*(\rho_{AB}) \geq \tilde{\eta}(\rho_{AB}, \{\Lambda_{1/\mu}(M_{a \vert x})\}, \{\rho_\lambda\})$.

\subsection{Conceptual comparison with existing methods}
The techniques described above lead to linear programs for computing lower and upper bounds for $\eta^*(\rho_{AB})$.
As we have mentioned, there are other methods available in the literature for computing the same quantities.
On the one hand, the method presented in Refs.~\cite{hirsch15, cavalcanti15} is the inspiration for the ones we propose, and in the previous sections we took it as the starting point for the presentation of our methods.
On the other hand, for two-qubit states under dichotomic measurements, the best known method to calculate these bounds also consists of a linear program which furthermore requires a polytope approximation of the Bloch sphere~\cite{Nguyen2019LHSCriterion}.
In this section, we discuss the conceptual relations between these three methods in further details.

Let us start by introducing a few definitions.
The \textit{set of steering outcomes} of a given state $\rho_{AB}$, denoted by $\Omega(\rho_{AB})$, is defined by
\begin{equation}
  \Omega(\rho_{AB}) = \{\Tr_A (M \otimes \1 ~\rho_{AB})~\vert~0 \leq M \leq \1\}.
\end{equation}
Additionally, given a probability distribution $\mu$ over the Bloch ball, we define its \textit{capacity} $\kappa(\mu)$ as
\begin{equation}
  \kappa(\mu) = \left\{\int d\mu(\sigma) g(\sigma) \sigma~\bigg\vert~0 \leq g(\sigma) \leq 1 ~\forall \sigma\right\},
\end{equation}
where the integral is carried over the Bloch ball.

By definition, a state $\rho_{AB}$ is unsteerable with respect to dichotomic measurements if there exists a probability distribution $\mu$ over the Bloch ball such that for every dichotomic measurement $x = \{M_{0 \vert x}, M_{1 \vert x}\}$ there are probability distributions $p(a \vert x, \sigma)$, such that
\begin{equation}
    \Tr_A (M_{a \vert x} \otimes \1 ~\rho_{AB}) = \int d\mu(\sigma) p(a \vert x, \sigma) \sigma,
\end{equation}
where $a = 0,1$ and the integral is once more carried over the Bloch ball.
That is, in terms of the definitions above, $\rho_{AB}$ is unsteerable with respect to dichotomic measurements if, and only if, there exists a probability distribution $\mu$ over the Bloch ball such that $\Omega(\rho_{AB}) \subset \kappa(\mu)$.

In a sense, every method to study the steerability of a quantum state under dichotomic measurements is tailored to verify the above inclusion \cite{nguyen2018}.
Firstly, the approach of Refs.~\cite{hirsch15, cavalcanti15} is to consider polytope approximations of the set $\Omega(\rho_{AB})$.
More precisely, to prove the unsteerability of $\rho_{AB}$, they construct an \textit{outer} polytope approximation of $\Omega(\rho_{AB})$, say, with $m$ vertices, and verify whether there exists a probability distribution $\mu$ over the Bloch ball such that each vertex of the outer polytope is included in $\kappa(\mu)$.
This can be done with a single SDP, whose size grows exponentially with $m$~\cite{hirsch15, cavalcanti15}.

Secondly, the method from Ref.~\cite{Nguyen2019LHSCriterion} instead considers polytope approximations of the sets $\kappa(\mu)$.
To show that $\rho_{AB}$ is unsteerable, the authors fix a set of $n$ states in the Bloch ball $\{\rho_\lambda\}$ and search for a probability distribution $\mu$ which only has nonzero values at the fixed states $\{\rho_\lambda\}$ and such that $\Omega(\rho_{AB}) \subset \kappa(\mu)$.
If $\rho_{AB}$ is a two-qubit state, and once the set of states $\{\rho_\lambda\}$ is fixed, this can be done with a single linear program, whose size grows with $n^3$~\cite{Nguyen2019LHSCriterion}.

Finally, in this work we consider \textit{both polytope approximations of} $\Omega(\rho_{AB})$ \textit{and polytope approximations of $\kappa(\mu)$}.
That is, to prove the unsteerability of $\rho_{AB}$, we fix an outer polytope approximation of $\Omega(\rho_{AB})$ with $m$ vertices, and we look for probability distributions $\mu$ restricted to a fixed set of $n$ states $\{\rho_\lambda\}$ such that $\kappa(\mu)$ includes all the $m$ vertices of the chosen polytope.
This results in a linear program whose size grows with the product $mn$.

In practice, the method developed in Ref.~\cite{Nguyen2019LHSCriterion} is better than the one from Refs.~\cite{hirsch15, cavalcanti15} for two-qubit quantum states.
In the following, we compare the former with the method proposed in this work.
In general, both methods provide quite similar bounds with similar computational resources, but while the method from Ref.~\cite{Nguyen2019LHSCriterion} seems to be advantageous for proving the unsteerability of quantum states, our method seems to be better for proving their steerability.

\subsection{Practical comparison with existing methods}
\subsubsection{Performance comparison}
To make a practical comparison between the methods we propose and the one from Ref.~\cite{Nguyen2019LHSCriterion} for generic two-qubit states, we computed lower and upper bounds on the steering robustness of 100 randomly generated steerable states using both techniques.
Such states were generated using the function \code{random_state} from the package \texttt{Ket.jl}~\cite{Ket.jl}, which uniformly samples a state according to the Hilbert-Schmidt metric, and a version of the linear program from Ref.~\cite{Nguyen2019LHSCriterion} to select only steerable states.
As the set of measurements to be used in the linear programs \eqref{eq:LP_steerability} and \eqref{eq:LP_unsteerability}, we considered the 406 measurements from Ref.~\cite{Designolle23Grothendieck}, and as the polytope approximation of the Bloch sphere, we used the same 614-vertex polytope considered in \cref{subsec:examplesJM_qubits}.
Meanwhile, for the linear program from Ref.~\cite{Nguyen2019LHSCriterion}, we used the 162-vertex polytope mentioned in \cref{tab:sloane_shr} as an approximation of the Bloch sphere.
The states we considered as well as the lower and upper bounds computed with each of the methods are available at \cite{gitcode}.

The measurements and polytopes just mentioned above were chosen in order to keep similar running times for both methods.
However, while the linear program from Ref.~\cite{Nguyen2019LHSCriterion} has roughly the same running time for every input state, we observe that our method's runtime varies significantly from state to state, ranging from a few seconds to a few minutes.
As we have mentioned before, we believe this is an inherent feature of our approach, which is related to the alignment between the polytope chosen to approximate the set of states, and the assemblage resulting from applying the chosen measurements on the input state.
Nevertheless, the time to compute the desired bounds for all the 100 considered states was roughly the same (around five hours) for both methods.

Interestingly, the bounds provided by our methods are quite similar to the ones obtained with the method from Ref.~\cite{Nguyen2019LHSCriterion}.
Regarding the lower bounds on the steering robustness, the average difference between the bounds computed by both methods is of order $10^{-4}$.
Meanwhile, for the upper bounds the average difference is of order $10^{-3}$.
This is a good indication of the quality of our methods, since to this day in the literature there is nothing comparable to the linear program in~\cite{Nguyen2019LHSCriterion}.

On the one hand, for 86 out of the 100 states, the method from Ref.~\cite{Nguyen2019LHSCriterion} provides better lower bounds on the steering robustness when compared to the linear program \eqref{eq:LP_unsteerability}.
On the other hand, for \textit{all} the studied states the upper bounds provided by the linear program \eqref{eq:LP_steerability} are better than the ones obtained with Ref.~\cite{Nguyen2019LHSCriterion}.
That is, the method from Ref.~\cite{Nguyen2019LHSCriterion} is generally better than ours to prove that quantum states are unsteerable, but the linear program \eqref{eq:LP_steerability} outperforms \cite{Nguyen2019LHSCriterion} in certifying the steerability of states.
Furthermore, the dual of the linear program \eqref{eq:LP_steerability} provides a steering inequality which can be violated by the studied state under the measurements previously chosen.
Thus, our method also provides a certificate of the steerability of the considered state, which is not directly granted by the method from~\cite{Nguyen2019LHSCriterion}.

\subsubsection{Flexibility over different noise models}
Another practical advantage of our approach is its flexibility with respect to different noise models on the states.
For a given two-qubit state $\rho_{AB}$, the linear program proposed in Ref.~\cite{Nguyen2019LHSCriterion} computes bounds on the maximum $\eta$ such that $\eta\rho_{AB} + (1-\eta)\frac{\1}{2} \otimes \rho_B$ is unsteerable, where $\rho_B = \Tr_A(\rho_{AB})$, as in \cref{eq:def_state_steering_rob}.
However, if one is interested in studying a different kind of noise model --- e.g.$~\eta\rho_{AB} + (1-\eta)\frac{\1}{4}$ --- one has to run the linear program from Ref.~\cite{Nguyen2019LHSCriterion} multiple times, using the bisection method.
Meanwhile, the methods we propose can be easily adapted to different noise models.
In particular, our method covers any noisy model where the noise parameter $\eta$ is linear in the given state $\rho_{AB}$.
It suffices to change \cref{eq:lp_steerability_lhs_constraint,eq:lp_unsteerability_lhs_constraint} according to the desired noise on the state, and a single run of each of the linear programs \eqref{eq:LP_steerability} and \eqref{eq:LP_unsteerability} provides tight bounds on the steering robustness with the considered noise.

\subsubsection{Higher dimensions and general POVMs}
The linear program from Ref.~\cite{Nguyen2019LHSCriterion} is limited to qudit-qubit states under dichotomic measurements, and cannot be straightforwardly extended to systems of higher dimension or general measurements (POVMs).
In comparison, the approach presented in this work has the advantage of being applicable to states of arbitrary dimension, and arbitrary POVMs.
Nonetheless, dealing with such general cases may be challenging, due to the necessity of finding good polytope approximations of the set of quantum states and measurements.
In this sense, with a realistic/reasonable execution time, the upper and lower bounds for higher dimensions may not be as tight as the ones we obtained for the two-qubit projective measurement scenario.
However, as presented in \cref{sec:examples_JM_qutrits} our methods do obtain non-trivial results for qutrits, and especially for proving the steerability of a state (see \cref{subsec:steerability}) our techniques are expected to perform well in the POVM case.
For instance, the results from \cref{sec:nonproj_qubit} can be viewed as steerability certificates for the two-qubit Werner state under POVMs.
Meanwhile, constructing LHS models for general POVMs may still be difficult, since our techniques require a good polytope approximation for the set of POVMs and the evaluation of its corresponding shrinking factor (see, e.g., Ref.~\cite{hirsch2018}).

\section{Conclusion}
In this work, we presented a method to quantify measurement incompatibility and steering whose complexity scales polynomially in the number of measurements, and which works for arbitrary POVMs in any dimension.
Such a scaling contrasts with standard techniques based on SDP, where the size of the optimization problem grows exponentially with the number of measurements.
More precisely, our approach trades off accuracy for efficiency, as its runtime also scales polynomially with the inverse of the desired precision $\epsilon^{-1}$, while the standard SDP approach scales logarithmically with $\epsilon^{-1}$.

In practice, while it is challenging to reach an arbitrarily small tolerance in high dimensions, in the case of qubits our methods provide remarkably tight bounds for the incompatibility robustness of considerably large sets of measurements with a very good performance.
For instance, we could evaluate the robustness of 400 projective qubit measurements with a four decimal digit precision in a minute timescale using a standard laptop, while previous methods could not even attempt to tackle scenarios with 25 measurements.
Similar performances can be achieved for sets of qubit POVMs.
For qutrit measurements, one can obtain non-trivial bounds on incompatibility robustnesses in scenarios that are intractable with the standard SDP, but which are not as tight as in the qubit case.
The Julia code we used for these and for all examples in this work is openly available and ready to use at \cite{gitcode}.

Additionally, our methods can be used to certify the steerability of quantum states, or in contrary, to construct LHS models for them under all possible local measurements.
Such techniques are applicable to states of any dimension, and their steerability can be studied with respect to generic POVMs.
In practice, for two-qubit states we found our results to be comparable with the state-of-the-art linear program developed in Ref.~\cite{Nguyen2019LHSCriterion}.
In our tests, our approach showed a better performance for certifying that a given state is steerable, and it had a performance similar to Ref.~\cite{Nguyen2019LHSCriterion} for constructing LHS models (slightly better in some cases, slightly worse in others).
We stress that the method from Ref.~\cite{Nguyen2019LHSCriterion} is restricted to qudit-qubit states, and it is tailored for a specific kind of noise on the state, thus becoming impractical if a different noise model is to be studied.
Meanwhile, our method is applicable to arbitrary linear noise model without change in performance, and works for arbitrary qudit-qudit states and general POVMs (but in some general cases, dealing with the polytope approximations may be challenging).
Finally, we believe other problems related to measurement incompatibility and steering may benefit from the methods developed in this work.

\section*{Acknowledgments}
We thank Ties-A.~Ohst and H.~Chau Nguyen for interesting discussions.
Work funded by QuantEdu France, a State aid managed by the French National Research Agency for France 2030 with the reference ANR-22-CMAS-0001.
MTQ is supported by the Agence Nationale de la Recherche (ANR) through the JCJC programme under grant number ANR-25-CE47-6396-01-HOQO-KS. 
Research reported in this paper was partially supported through the Research Campus Modal funded by the German Federal Ministry of Education and Research (fund numbers 05M14ZAM, 05M20ZBM) and the Deutsche Forschungsgemeinschaft (DFG) through the DFG Cluster of Excellence MATH+.

\nocite{apsrev42Control}

\bibliographystyle{0_MTQ_apsrev4-2_corrected}
\bibliography{LP_for_JM.bib}

\begin{thebibliography}{75}%
\makeatletter
\providecommand \@ifxundefined [1]{%
 \@ifx{#1\undefined}
}%
\providecommand \@ifnum [1]{%
 \ifnum #1\expandafter \@firstoftwo
 \else \expandafter \@secondoftwo
 \fi
}%
\providecommand \@ifx [1]{%
 \ifx #1\expandafter \@firstoftwo
 \else \expandafter \@secondoftwo
 \fi
}%
\providecommand \natexlab [1]{#1}%
\providecommand \enquote  [1]{``#1''}%
\providecommand \bibnamefont  [1]{#1}%
\providecommand \bibfnamefont [1]{#1}%
\providecommand \citenamefont [1]{#1}%
\providecommand \href@noop [0]{\@secondoftwo}%
\providecommand \href [0]{\begingroup \@sanitize@url \@href}%
\providecommand \@href[1]{\@@startlink{#1}\@@href}%
\providecommand \@@href[1]{\endgroup#1\@@endlink}%
\providecommand \@sanitize@url [0]{\catcode `\\12\catcode `\$12\catcode `\&12\catcode `\#12\catcode `\^12\catcode `\_12\catcode `\%12\relax}%
\providecommand \@@startlink[1]{}%
\providecommand \@@endlink[0]{}%
\providecommand \url  [0]{\begingroup\@sanitize@url \@url }%
\providecommand \@url [1]{\endgroup\@href {#1}{\urlprefix }}%
\providecommand \urlprefix  [0]{URL }%
\providecommand \Eprint [0]{\href }%
\providecommand \doibase [0]{https://doi.org/}%
\providecommand \selectlanguage [0]{\@gobble}%
\providecommand \bibinfo  [0]{\@secondoftwo}%
\providecommand \bibfield  [0]{\@secondoftwo}%
\providecommand \translation [1]{[#1]}%
\providecommand \BibitemOpen [0]{}%
\providecommand \bibitemStop [0]{}%
\providecommand \bibitemNoStop [0]{.\EOS\space}%
\providecommand \EOS [0]{\spacefactor3000\relax}%
\providecommand \BibitemShut  [1]{\csname bibitem#1\endcsname}%
\let\auto@bib@innerbib\@empty
\bibitem [{\citenamefont {G\"uhne}\ \emph {et~al.}(2023)\citenamefont {G\"uhne}, \citenamefont {Haapasalo}, \citenamefont {Kraft}, \citenamefont {Pellonp\"a\"a},\ and\ \citenamefont {Uola}}]{Guhne2023IncompReview}%
  \BibitemOpen
  \bibfield  {author} {\bibinfo {author} {\bibfnamefont {O.}~\bibnamefont {G\"uhne}}, \bibinfo {author} {\bibfnamefont {E.}~\bibnamefont {Haapasalo}}, \bibinfo {author} {\bibfnamefont {T.}~\bibnamefont {Kraft}}, \bibinfo {author} {\bibfnamefont {J.-P.}\ \bibnamefont {Pellonp\"a\"a}},\ and\ \bibinfo {author} {\bibfnamefont {R.}~\bibnamefont {Uola}},\ }\bibfield  {title} {\bibinfo {title} {Colloquium: Incompatible measurements in quantum information science},\ }\href {https://doi.org/10.1103/RevModPhys.95.011003} {\bibfield  {journal} {\bibinfo  {journal} {Rev. Mod. Phys.}\ }\textbf {\bibinfo {volume} {95}},\ \bibinfo {pages} {011003} (\bibinfo {year} {2023})},\ \Eprint {https://arxiv.org/abs/2112.06784}{arXiv:2112.06784 [quant-ph]}\BibitemShut {NoStop}%
\bibitem [{\citenamefont {Bell}(1964)}]{bell64}%
  \BibitemOpen
  \bibfield  {author} {\bibinfo {author} {\bibfnamefont {J.~S.}\ \bibnamefont {Bell}},\ }\bibfield  {title} {\bibinfo {title} {{On the Einstein-Poldolsky-Rosen paradox}},\ }\href {https://cds.cern.ch/record/111654/files/vol1p195-200_001.pdf} {\bibfield  {journal} {\bibinfo  {journal} {Physics}\ }\textbf {\bibinfo {volume} {1}},\ \bibinfo {pages} {195--200} (\bibinfo {year} {1964})}\BibitemShut {NoStop}%
\bibitem [{\citenamefont {{Brunner}}\ \emph {et~al.}(2014)\citenamefont {{Brunner}}, \citenamefont {{Cavalcanti}}, \citenamefont {{Pironio}}, \citenamefont {{Scarani}},\ and\ \citenamefont {{Wehner}}}]{brunner_review}%
  \BibitemOpen
  \bibfield  {author} {\bibinfo {author} {\bibfnamefont {N.}~\bibnamefont {{Brunner}}}, \bibinfo {author} {\bibfnamefont {D.}~\bibnamefont {{Cavalcanti}}}, \bibinfo {author} {\bibfnamefont {S.}~\bibnamefont {{Pironio}}}, \bibinfo {author} {\bibfnamefont {V.}~\bibnamefont {{Scarani}}},\ and\ \bibinfo {author} {\bibfnamefont {S.}~\bibnamefont {{Wehner}}},\ }\bibfield  {title} {\bibinfo {title} {{Bell nonlocality}},\ }\href {https://doi.org/10.1103/RevModPhys.86.419} {\bibfield  {journal} {\bibinfo  {journal} {Rev. Mod. Phys.}\ }\textbf {\bibinfo {volume} {86}},\ \bibinfo {pages} {419--478} (\bibinfo {year} {2014})},\ \Eprint {https://arxiv.org/abs/1303.2849}{arXiv:1303.2849 [quant-ph]}\BibitemShut {NoStop}%
\bibitem [{\citenamefont {Fine}(1982)}]{fine82}%
  \BibitemOpen
  \bibfield  {author} {\bibinfo {author} {\bibfnamefont {A.}~\bibnamefont {Fine}},\ }\bibfield  {title} {\bibinfo {title} {{Hidden Variables, Joint Probability, and the Bell Inequalities}},\ }\href {https://doi.org/10.1103/PhysRevLett.48.291} {\bibfield  {journal} {\bibinfo  {journal} {Phys. Rev. Lett.}\ }\textbf {\bibinfo {volume} {48}},\ \bibinfo {pages} {291--295} (\bibinfo {year} {1982})}\BibitemShut {NoStop}%
\bibitem [{\citenamefont {Quintino}\ \emph {et~al.}(2016)\citenamefont {Quintino}, \citenamefont {Bowles}, \citenamefont {Hirsch},\ and\ \citenamefont {Brunner}}]{quintino15Incompatible}%
  \BibitemOpen
  \bibfield  {author} {\bibinfo {author} {\bibfnamefont {M.~T.}\ \bibnamefont {Quintino}}, \bibinfo {author} {\bibfnamefont {J.}~\bibnamefont {Bowles}}, \bibinfo {author} {\bibfnamefont {F.}~\bibnamefont {Hirsch}},\ and\ \bibinfo {author} {\bibfnamefont {N.}~\bibnamefont {Brunner}},\ }\bibfield  {title} {\bibinfo {title} {Incompatible quantum measurements admitting a local-hidden-variable model},\ }\href {https://doi.org/10.1103/PhysRevA.93.052115} {\bibfield  {journal} {\bibinfo  {journal} {Phys. Rev. A}\ }\textbf {\bibinfo {volume} {93}},\ \bibinfo {pages} {052115} (\bibinfo {year} {2016})},\ \Eprint {https://arxiv.org/abs/1406.6976}{arXiv:1406.6976 [quant-ph]}\BibitemShut {NoStop}%
\bibitem [{\citenamefont {{Hirsch}}\ \emph {et~al.}(2018)\citenamefont {{Hirsch}}, \citenamefont {{Quintino}},\ and\ \citenamefont {{Brunner}}}]{hirsch2018}%
  \BibitemOpen
  \bibfield  {author} {\bibinfo {author} {\bibfnamefont {F.}~\bibnamefont {{Hirsch}}}, \bibinfo {author} {\bibfnamefont {M.~T.}\ \bibnamefont {{Quintino}}},\ and\ \bibinfo {author} {\bibfnamefont {N.}~\bibnamefont {{Brunner}}},\ }\bibfield  {title} {\bibinfo {title} {{Quantum measurement incompatibility does not imply Bell nonlocality}},\ }\href {https://doi.org/10.1103/PhysRevA.97.012129} {\bibfield  {journal} {\bibinfo  {journal} {Phys. Rev. A}\ }\textbf {\bibinfo {volume} {97}},\ \bibinfo {eid} {012129} (\bibinfo {year} {2018})},\ \Eprint {https://arxiv.org/abs/1707.06960}{arXiv:1707.06960 [quant-ph]}\BibitemShut {NoStop}%
\bibitem [{\citenamefont {{Bene}}\ and\ \citenamefont {{V{\'e}rtesi}}(2018)}]{bene2018}%
  \BibitemOpen
  \bibfield  {author} {\bibinfo {author} {\bibfnamefont {E.}~\bibnamefont {{Bene}}}\ and\ \bibinfo {author} {\bibfnamefont {T.}~\bibnamefont {{V{\'e}rtesi}}},\ }\bibfield  {title} {\bibinfo {title} {{Measurement incompatibility does not give rise to Bell violation in general}},\ }\href {https://doi.org/10.1088/1367-2630/aa9ca3} {\bibfield  {journal} {\bibinfo  {journal} {New J. Phys.}\ }\textbf {\bibinfo {volume} {20}},\ \bibinfo {eid} {013021} (\bibinfo {year} {2018})},\ \Eprint {https://arxiv.org/abs/1705.10069}{arXiv:1705.10069 [quant-ph]}\BibitemShut {NoStop}%
\bibitem [{\citenamefont {Pl\'avala}\ \emph {et~al.}(2025)\citenamefont {Pl\'avala}, \citenamefont {G\"uhne},\ and\ \citenamefont {Quintino}}]{Plavala2024incomp}%
  \BibitemOpen
  \bibfield  {author} {\bibinfo {author} {\bibfnamefont {M.}~\bibnamefont {Pl\'avala}}, \bibinfo {author} {\bibfnamefont {O.}~\bibnamefont {G\"uhne}},\ and\ \bibinfo {author} {\bibfnamefont {M.~T.}\ \bibnamefont {Quintino}},\ }\bibfield  {title} {\bibinfo {title} {{All Incompatible Measurements on Qubits Lead to Multiparticle {Bell} Nonlocality}},\ }\href {https://doi.org/10.1103/PhysRevLett.134.200201} {\bibfield  {journal} {\bibinfo  {journal} {Phys. Rev. Lett.}\ }\textbf {\bibinfo {volume} {134}},\ \bibinfo {pages} {200201} (\bibinfo {year} {2025})},\ \Eprint {https://arxiv.org/abs/2403.10564}{arXiv:2403.10564 [quant-ph]}\BibitemShut {NoStop}%
\bibitem [{\citenamefont {{Quintino}}\ \emph {et~al.}(2014)\citenamefont {{Quintino}}, \citenamefont {{V{\'e}rtesi}},\ and\ \citenamefont {{Brunner}}}]{quintino14}%
  \BibitemOpen
  \bibfield  {author} {\bibinfo {author} {\bibfnamefont {M.~T.}\ \bibnamefont {{Quintino}}}, \bibinfo {author} {\bibfnamefont {T.}~\bibnamefont {{V{\'e}rtesi}}},\ and\ \bibinfo {author} {\bibfnamefont {N.}~\bibnamefont {{Brunner}}},\ }\bibfield  {title} {\bibinfo {title} {{Joint Measurability, Einstein-Podolsky-Rosen Steering, and Bell Nonlocality}},\ }\href {https://doi.org/10.1103/PhysRevLett.113.160402} {\bibfield  {journal} {\bibinfo  {journal} {Phys. Rev. Lett.}\ }\textbf {\bibinfo {volume} {113}},\ \bibinfo {eid} {160402} (\bibinfo {year} {2014})},\ \Eprint {https://arxiv.org/abs/1406.6976}{arXiv:1406.6976 [quant-ph]}\BibitemShut {NoStop}%
\bibitem [{\citenamefont {{Uola}}\ \emph {et~al.}(2014)\citenamefont {{Uola}}, \citenamefont {{Moroder}},\ and\ \citenamefont {{G{\"u}hne}}}]{uola14}%
  \BibitemOpen
  \bibfield  {author} {\bibinfo {author} {\bibfnamefont {R.}~\bibnamefont {{Uola}}}, \bibinfo {author} {\bibfnamefont {T.}~\bibnamefont {{Moroder}}},\ and\ \bibinfo {author} {\bibfnamefont {O.}~\bibnamefont {{G{\"u}hne}}},\ }\bibfield  {title} {\bibinfo {title} {{Joint Measurability of Generalized Measurements Implies Classicality}},\ }\href {https://doi.org/10.1103/PhysRevLett.113.160403} {\bibfield  {journal} {\bibinfo  {journal} {Phys. Rev. Lett.}\ }\textbf {\bibinfo {volume} {113}},\ \bibinfo {eid} {160403} (\bibinfo {year} {2014})},\ \Eprint {https://arxiv.org/abs/1407.2224}{arXiv:1407.2224 [quant-ph]}\BibitemShut {NoStop}%
\bibitem [{\citenamefont {{Uola}}\ \emph {et~al.}(2020)\citenamefont {{Uola}}, \citenamefont {{Costa}}, \citenamefont {{Nguyen}},\ and\ \citenamefont {{G{\"u}hne}}}]{Uola2020SteeringReview}%
  \BibitemOpen
  \bibfield  {author} {\bibinfo {author} {\bibfnamefont {R.}~\bibnamefont {{Uola}}}, \bibinfo {author} {\bibfnamefont {A.~C.~S.}\ \bibnamefont {{Costa}}}, \bibinfo {author} {\bibfnamefont {H.~C.}\ \bibnamefont {{Nguyen}}},\ and\ \bibinfo {author} {\bibfnamefont {O.}~\bibnamefont {{G{\"u}hne}}},\ }\bibfield  {title} {\bibinfo {title} {{Quantum steering}},\ }\href {https://doi.org/10.1103/RevModPhys.92.015001} {\bibfield  {journal} {\bibinfo  {journal} {Rev. Mod. Phys.}\ }\textbf {\bibinfo {volume} {92}},\ \bibinfo {eid} {015001} (\bibinfo {year} {2020})},\ \Eprint {https://arxiv.org/abs/1903.06663}{arXiv:1903.06663 [quant-ph]}\BibitemShut {NoStop}%
\bibitem [{\citenamefont {{Buscemi}}\ \emph {et~al.}(2020)\citenamefont {{Buscemi}}, \citenamefont {{Chitambar}},\ and\ \citenamefont {{Zhou}}}]{Buscemi2020PAM}%
  \BibitemOpen
  \bibfield  {author} {\bibinfo {author} {\bibfnamefont {F.}~\bibnamefont {{Buscemi}}}, \bibinfo {author} {\bibfnamefont {E.}~\bibnamefont {{Chitambar}}},\ and\ \bibinfo {author} {\bibfnamefont {W.}~\bibnamefont {{Zhou}}},\ }\bibfield  {title} {\bibinfo {title} {{Complete Resource Theory of Quantum Incompatibility as Quantum Programmability}},\ }\href {https://doi.org/10.1103/PhysRevLett.124.120401} {\bibfield  {journal} {\bibinfo  {journal} {Phys. Rev. Lett.}\ }\textbf {\bibinfo {volume} {124}},\ \bibinfo {eid} {120401} (\bibinfo {year} {2020})},\ \Eprint {https://arxiv.org/abs/1908.11274}{arXiv:1908.11274 [quant-ph]}\BibitemShut {NoStop}%
\bibitem [{\citenamefont {{Carmeli}}\ \emph {et~al.}(2019)\citenamefont {{Carmeli}}, \citenamefont {{Heinosaari}},\ and\ \citenamefont {{Toigo}}}]{Carmeli2019PAM}%
  \BibitemOpen
  \bibfield  {author} {\bibinfo {author} {\bibfnamefont {C.}~\bibnamefont {{Carmeli}}}, \bibinfo {author} {\bibfnamefont {T.}~\bibnamefont {{Heinosaari}}},\ and\ \bibinfo {author} {\bibfnamefont {A.}~\bibnamefont {{Toigo}}},\ }\bibfield  {title} {\bibinfo {title} {{Quantum Incompatibility Witnesses}},\ }\href {https://doi.org/10.1103/PhysRevLett.122.130402} {\bibfield  {journal} {\bibinfo  {journal} {Phys. Rev. Lett.}\ }\textbf {\bibinfo {volume} {122}},\ \bibinfo {eid} {130402} (\bibinfo {year} {2019})},\ \Eprint {https://arxiv.org/abs/1812.02985}{arXiv:1812.02985 [quant-ph]}\BibitemShut {NoStop}%
\bibitem [{\citenamefont {{Guerini}}\ \emph {et~al.}(2019)\citenamefont {{Guerini}}, \citenamefont {{Quintino}},\ and\ \citenamefont {{Aolita}}}]{Guerini2019PAM}%
  \BibitemOpen
  \bibfield  {author} {\bibinfo {author} {\bibfnamefont {L.}~\bibnamefont {{Guerini}}}, \bibinfo {author} {\bibfnamefont {M.~T.}\ \bibnamefont {{Quintino}}},\ and\ \bibinfo {author} {\bibfnamefont {L.}~\bibnamefont {{Aolita}}},\ }\bibfield  {title} {\bibinfo {title} {{Distributed sampling, quantum communication witnesses, and measurement incompatibility}},\ }\href {https://doi.org/10.1103/PhysRevA.100.042308} {\bibfield  {journal} {\bibinfo  {journal} {Phys. Rev. A}\ }\textbf {\bibinfo {volume} {100}},\ \bibinfo {eid} {042308} (\bibinfo {year} {2019})},\ \Eprint {https://arxiv.org/abs/1904.08435}{arXiv:1904.08435 [quant-ph]}\BibitemShut {NoStop}%
\bibitem [{\citenamefont {{Ac{\'{\i}}n}}\ \emph {et~al.}(2007)\citenamefont {{Ac{\'{\i}}n}}, \citenamefont {{Brunner}}, \citenamefont {{Gisin}}, \citenamefont {{Massar}}, \citenamefont {{Pironio}},\ and\ \citenamefont {{Scarani}}}]{acin07}%
  \BibitemOpen
  \bibfield  {author} {\bibinfo {author} {\bibfnamefont {A.}~\bibnamefont {{Ac{\'{\i}}n}}}, \bibinfo {author} {\bibfnamefont {N.}~\bibnamefont {{Brunner}}}, \bibinfo {author} {\bibfnamefont {N.}~\bibnamefont {{Gisin}}}, \bibinfo {author} {\bibfnamefont {S.}~\bibnamefont {{Massar}}}, \bibinfo {author} {\bibfnamefont {S.}~\bibnamefont {{Pironio}}},\ and\ \bibinfo {author} {\bibfnamefont {V.}~\bibnamefont {{Scarani}}},\ }\bibfield  {title} {\bibinfo {title} {{Device-Independent Security of Quantum Cryptography against Collective Attacks}},\ }\href {https://doi.org/10.1103/PhysRevLett.98.230501} {\bibfield  {journal} {\bibinfo  {journal} {Phys. Rev. Lett.}\ }\textbf {\bibinfo {volume} {98}},\ \bibinfo {eid} {230501} (\bibinfo {year} {2007})},\ \Eprint {https://arxiv.org/abs/1705.10069}{arXiv:1705.10069 [quant-ph]}\BibitemShut {NoStop}%
\bibitem [{\citenamefont {{Skrzypczyk}}\ \emph {et~al.}(2019)\citenamefont {{Skrzypczyk}}, \citenamefont {{{\v{S}}upi{\'c}}},\ and\ \citenamefont {{Cavalcanti}}}]{Paul2019PAM}%
  \BibitemOpen
  \bibfield  {author} {\bibinfo {author} {\bibfnamefont {P.}~\bibnamefont {{Skrzypczyk}}}, \bibinfo {author} {\bibfnamefont {I.}~\bibnamefont {{{\v{S}}upi{\'c}}}},\ and\ \bibinfo {author} {\bibfnamefont {D.}~\bibnamefont {{Cavalcanti}}},\ }\bibfield  {title} {\bibinfo {title} {{All Sets of Incompatible Measurements give an Advantage in Quantum State Discrimination}},\ }\href {https://doi.org/10.1103/PhysRevLett.122.130403} {\bibfield  {journal} {\bibinfo  {journal} {Phys. Rev. Lett.}\ }\textbf {\bibinfo {volume} {122}},\ \bibinfo {eid} {130403} (\bibinfo {year} {2019})},\ \Eprint {https://arxiv.org/abs/1901.00816}{arXiv:1901.00816 [quant-ph]}\BibitemShut {NoStop}%
\bibitem [{\citenamefont {{Carmeli}}\ \emph {et~al.}(2018)\citenamefont {{Carmeli}}, \citenamefont {{Heinosaari}},\ and\ \citenamefont {{Toigo}}}]{carmeli2018}%
  \BibitemOpen
  \bibfield  {author} {\bibinfo {author} {\bibfnamefont {C.}~\bibnamefont {{Carmeli}}}, \bibinfo {author} {\bibfnamefont {T.}~\bibnamefont {{Heinosaari}}},\ and\ \bibinfo {author} {\bibfnamefont {A.}~\bibnamefont {{Toigo}}},\ }\bibfield  {title} {\bibinfo {title} {{State discrimination with postmeasurement information and incompatibility of quantum measurements}},\ }\href {https://doi.org/10.1103/PhysRevA.98.012126} {\bibfield  {journal} {\bibinfo  {journal} {Phys. Rev. A}\ }\textbf {\bibinfo {volume} {98}},\ \bibinfo {eid} {012126} (\bibinfo {year} {2018})},\ \Eprint {https://arxiv.org/abs/1804.09693}{arXiv:1804.09693 [quant-ph]}\BibitemShut {NoStop}%
\bibitem [{\citenamefont {Schumacher}\ and\ \citenamefont {Westmoreland}(2010)}]{Schumacher2010textbook}%
  \BibitemOpen
  \bibfield  {author} {\bibinfo {author} {\bibfnamefont {B.}~\bibnamefont {Schumacher}}\ and\ \bibinfo {author} {\bibfnamefont {M.}~\bibnamefont {Westmoreland}},\ }\href {https://doi.org/10.1017/CBO9780511814006} {\bibinfo {title} {Quantum Processes Systems, and Information}}\ (\bibinfo  {publisher} {Cambridge University Press},\ \bibinfo {year} {2010})\BibitemShut {NoStop}%
\bibitem [{\citenamefont {Busch}(1986)}]{busch1986}%
  \BibitemOpen
  \bibfield  {author} {\bibinfo {author} {\bibfnamefont {P.}~\bibnamefont {Busch}},\ }\bibfield  {title} {\bibinfo {title} {Unsharp reality and joint measurements for spin observables},\ }\href {https://doi.org/10.1103/PhysRevD.33.2253} {\bibfield  {journal} {\bibinfo  {journal} {Phys. Rev. D}\ }\textbf {\bibinfo {volume} {33}},\ \bibinfo {pages} {2253--2261} (\bibinfo {year} {1986})}\BibitemShut {NoStop}%
\bibitem [{\citenamefont {{Yu}}\ \emph {et~al.}(2008)\citenamefont {{Yu}}, \citenamefont {{Liu}}, \citenamefont {{Li}},\ and\ \citenamefont {{Oh}}}]{yu2008biased}%
  \BibitemOpen
  \bibfield  {author} {\bibinfo {author} {\bibfnamefont {S.}~\bibnamefont {{Yu}}}, \bibinfo {author} {\bibfnamefont {N.}~\bibnamefont {{Liu}}}, \bibinfo {author} {\bibfnamefont {L.}~\bibnamefont {{Li}}},\ and\ \bibinfo {author} {\bibfnamefont {C.~H.}\ \bibnamefont {{Oh}}},\ }\bibfield  {title} {\bibinfo {title} {{Joint measurement of two unsharp observables of a qubit}},\ }\href {https://arxiv.org/abs/0805.1538} {\bibfield  {journal} {\bibinfo  {journal} {arXiv:0805.1538}\ } (\bibinfo {year} {2008})}\BibitemShut {NoStop}%
\bibitem [{\citenamefont {{Busch}}\ and\ \citenamefont {{Schmidt}}(2010)}]{busch2010biased}%
  \BibitemOpen
  \bibfield  {author} {\bibinfo {author} {\bibfnamefont {P.}~\bibnamefont {{Busch}}}\ and\ \bibinfo {author} {\bibfnamefont {H.-J.}\ \bibnamefont {{Schmidt}}},\ }\bibfield  {title} {\bibinfo {title} {{Coexistence of qubit effects}},\ }\href {https://doi.org/10.1007/s11128-009-0109-x} {\bibfield  {journal} {\bibinfo  {journal} {Quantum Information Processing}\ }\textbf {\bibinfo {volume} {9}},\ \bibinfo {pages} {143--169} (\bibinfo {year} {2010})},\ \Eprint {https://arxiv.org/abs/0802.4167}{arXiv:0802.4167 [quant-ph]}\BibitemShut {NoStop}%
\bibitem [{\citenamefont {{Stano}}\ \emph {et~al.}(2008)\citenamefont {{Stano}}, \citenamefont {{Reitzner}},\ and\ \citenamefont {{Heinosaari}}}]{stano2010biased}%
  \BibitemOpen
  \bibfield  {author} {\bibinfo {author} {\bibfnamefont {P.}~\bibnamefont {{Stano}}}, \bibinfo {author} {\bibfnamefont {D.}~\bibnamefont {{Reitzner}}},\ and\ \bibinfo {author} {\bibfnamefont {T.}~\bibnamefont {{Heinosaari}}},\ }\bibfield  {title} {\bibinfo {title} {{Coexistence of qubit effects}},\ }\href {https://doi.org/10.1103/PhysRevA.78.012315} {\bibfield  {journal} {\bibinfo  {journal} {Phys. Rev. A}\ }\textbf {\bibinfo {volume} {78}},\ \bibinfo {eid} {012315} (\bibinfo {year} {2008})},\ \Eprint {https://arxiv.org/abs/0802.4248}{arXiv:0802.4248 [quant-ph]}\BibitemShut {NoStop}%
\bibitem [{\citenamefont {{Pal}}\ and\ \citenamefont {{Ghosh}}(2011)}]{pal2011}%
  \BibitemOpen
  \bibfield  {author} {\bibinfo {author} {\bibfnamefont {R.}~\bibnamefont {{Pal}}}\ and\ \bibinfo {author} {\bibfnamefont {S.}~\bibnamefont {{Ghosh}}},\ }\bibfield  {title} {\bibinfo {title} {{Approximate joint measurement of qubit observables through an Arthur-Kelly model}},\ }\href {https://doi.org/10.1088/1751-8113/44/48/485303} {\bibfield  {journal} {\bibinfo  {journal} {J. Phys. A: Math. Theor.}\ }\textbf {\bibinfo {volume} {44}},\ \bibinfo {eid} {485303} (\bibinfo {year} {2011})},\ \Eprint {https://arxiv.org/abs/1010.2878}{arXiv:1010.2878 [quant-ph]}\BibitemShut {NoStop}%
\bibitem [{\citenamefont {{Yu}}\ and\ \citenamefont {{Oh}}(2013)}]{yu2013}%
  \BibitemOpen
  \bibfield  {author} {\bibinfo {author} {\bibfnamefont {S.}~\bibnamefont {{Yu}}}\ and\ \bibinfo {author} {\bibfnamefont {C.~H.}\ \bibnamefont {{Oh}}},\ }\bibfield  {title} {\bibinfo {title} {{Quantum contextuality and joint measurement of three observables of a qubit}},\ }\href {https://arxiv.org/abs/1312.6470} {\bibfield  {journal} {\bibinfo  {journal} {arXiv:1312.6470}\ } (\bibinfo {year} {2013})}\BibitemShut {NoStop}%
\bibitem [{\citenamefont {{Grinko}}\ and\ \citenamefont {{Uola}}(2025)}]{grinko2024}%
  \BibitemOpen
  \bibfield  {author} {\bibinfo {author} {\bibfnamefont {D.}~\bibnamefont {{Grinko}}}\ and\ \bibinfo {author} {\bibfnamefont {R.}~\bibnamefont {{Uola}}},\ }\bibfield  {title} {\bibinfo {title} {{Compatibility of Binary Qubit Measurements}},\ }\href {https://doi.org/10.1103/vv1h-5mf9} {\bibfield  {journal} {\bibinfo  {journal} {Phys. Rev. Lett.}\ }\textbf {\bibinfo {volume} {135}},\ \bibinfo {eid} {200201} (\bibinfo {year} {2025})},\ \Eprint {https://arxiv.org/abs/2407.07711}{arXiv:2407.07711 [quant-ph]}\BibitemShut {NoStop}%
\bibitem [{\citenamefont {{Wolf}}\ \emph {et~al.}(2009)\citenamefont {{Wolf}}, \citenamefont {{Perez-Garcia}},\ and\ \citenamefont {{Fernandez}}}]{wolf09}%
  \BibitemOpen
  \bibfield  {author} {\bibinfo {author} {\bibfnamefont {M.~M.}\ \bibnamefont {{Wolf}}}, \bibinfo {author} {\bibfnamefont {D.}~\bibnamefont {{Perez-Garcia}}},\ and\ \bibinfo {author} {\bibfnamefont {C.}~\bibnamefont {{Fernandez}}},\ }\bibfield  {title} {\bibinfo {title} {{Measurements Incompatible in Quantum Theory Cannot Be Measured Jointly in Any Other No-Signaling Theory}},\ }\href {https://doi.org/10.1103/PhysRevLett.103.230402} {\bibfield  {journal} {\bibinfo  {journal} {Phys. Rev. Lett.}\ }\textbf {\bibinfo {volume} {103}},\ \bibinfo {eid} {230402} (\bibinfo {year} {2009})},\ \Eprint {https://arxiv.org/abs/0905.2998}{arXiv:0905.2998 [quant-ph]}\BibitemShut {NoStop}%
\bibitem [{\citenamefont {{Cavalcanti}}\ and\ \citenamefont {{Skrzypczyk}}(2016)}]{cavalcanti2016quantifiers}%
  \BibitemOpen
  \bibfield  {author} {\bibinfo {author} {\bibfnamefont {D.}~\bibnamefont {{Cavalcanti}}}\ and\ \bibinfo {author} {\bibfnamefont {P.}~\bibnamefont {{Skrzypczyk}}},\ }\bibfield  {title} {\bibinfo {title} {{Quantitative relations between measurement incompatibility, quantum steering, and nonlocality}},\ }\href {https://doi.org/10.1103/PhysRevA.93.052112} {\bibfield  {journal} {\bibinfo  {journal} {Phys. Rev. A}\ }\textbf {\bibinfo {volume} {93}},\ \bibinfo {eid} {052112} (\bibinfo {year} {2016})},\ \Eprint {https://arxiv.org/abs/1601.07450}{arXiv:1601.07450 [quant-ph]}\BibitemShut {NoStop}%
\bibitem [{\citenamefont {{Designolle}}\ \emph {et~al.}(2019)\citenamefont {{Designolle}}, \citenamefont {{Farkas}},\ and\ \citenamefont {{Kaniewski}}}]{designolle2019quantifiers}%
  \BibitemOpen
  \bibfield  {author} {\bibinfo {author} {\bibfnamefont {S.}~\bibnamefont {{Designolle}}}, \bibinfo {author} {\bibfnamefont {M.}~\bibnamefont {{Farkas}}},\ and\ \bibinfo {author} {\bibfnamefont {J.}~\bibnamefont {{Kaniewski}}},\ }\bibfield  {title} {\bibinfo {title} {{Incompatibility robustness of quantum measurements: a unified framework}},\ }\href {https://doi.org/10.1088/1367-2630/ab5020} {\bibfield  {journal} {\bibinfo  {journal} {New J. Phys.}\ }\textbf {\bibinfo {volume} {21}},\ \bibinfo {eid} {113053} (\bibinfo {year} {2019})},\ \Eprint {https://arxiv.org/abs/1906.00448}{arXiv:1906.00448 [quant-ph]}\BibitemShut {NoStop}%
\bibitem [{\citenamefont {Fawzi}(2021)}]{Fawzi18}%
  \BibitemOpen
  \bibfield  {author} {\bibinfo {author} {\bibfnamefont {H.}~\bibnamefont {Fawzi}},\ }\bibfield  {title} {\bibinfo {title} {{On Polyhedral Approximations of the Positive Semidefinite Cone}},\ }\href {https://doi.org/10.1287/moor.2020.1077} {\bibfield  {journal} {\bibinfo  {journal} {Math. Oper. Res.}\ }\textbf {\bibinfo {volume} {46}},\ \bibinfo {pages} {1479--1489} (\bibinfo {year} {2021})},\ \Eprint {https://arxiv.org/abs/2003.00785}{arXiv:2003.00785}\BibitemShut {NoStop}%
\bibitem [{\citenamefont {{Hirsch}}\ \emph {et~al.}(2016)\citenamefont {{Hirsch}}, \citenamefont {{Quintino}}, \citenamefont {{V{\'e}rtesi}}, \citenamefont {{Pusey}},\ and\ \citenamefont {{Brunner}}}]{hirsch15}%
  \BibitemOpen
  \bibfield  {author} {\bibinfo {author} {\bibfnamefont {F.}~\bibnamefont {{Hirsch}}}, \bibinfo {author} {\bibfnamefont {M.~T.}\ \bibnamefont {{Quintino}}}, \bibinfo {author} {\bibfnamefont {T.}~\bibnamefont {{V{\'e}rtesi}}}, \bibinfo {author} {\bibfnamefont {M.~F.}\ \bibnamefont {{Pusey}}},\ and\ \bibinfo {author} {\bibfnamefont {N.}~\bibnamefont {{Brunner}}},\ }\bibfield  {title} {\bibinfo {title} {{Algorithmic Construction of Local Hidden Variable Models for Entangled Quantum States}},\ }\href {https://doi.org/10.1103/PhysRevLett.117.190402} {\bibfield  {journal} {\bibinfo  {journal} {Phys. Rev. Lett.}\ }\textbf {\bibinfo {volume} {117}},\ \bibinfo {eid} {190402} (\bibinfo {year} {2016})},\ \Eprint {https://arxiv.org/abs/1512.00262}{arXiv:1512.00262 [quant-ph]}\BibitemShut {NoStop}%
\bibitem [{\citenamefont {{Cavalcanti}}\ \emph {et~al.}(2016)\citenamefont {{Cavalcanti}}, \citenamefont {{Guerini}}, \citenamefont {{Rabelo}},\ and\ \citenamefont {{Skrzypczyk}}}]{cavalcanti15}%
  \BibitemOpen
  \bibfield  {author} {\bibinfo {author} {\bibfnamefont {D.}~\bibnamefont {{Cavalcanti}}}, \bibinfo {author} {\bibfnamefont {L.}~\bibnamefont {{Guerini}}}, \bibinfo {author} {\bibfnamefont {R.}~\bibnamefont {{Rabelo}}},\ and\ \bibinfo {author} {\bibfnamefont {P.}~\bibnamefont {{Skrzypczyk}}},\ }\bibfield  {title} {\bibinfo {title} {{General Method for Constructing Local Hidden Variable Models for Entangled Quantum States}},\ }\href {https://doi.org/10.1103/PhysRevLett.117.190401} {\bibfield  {journal} {\bibinfo  {journal} {Phys. Rev. Lett.}\ }\textbf {\bibinfo {volume} {117}},\ \bibinfo {eid} {190401} (\bibinfo {year} {2016})},\ \Eprint {https://arxiv.org/abs/1512.00277}{arXiv:1512.00277 [quant-ph]}\BibitemShut {NoStop}%
\bibitem [{\citenamefont {Nguyen}\ \emph {et~al.}(2019)\citenamefont {Nguyen}, \citenamefont {Nguyen},\ and\ \citenamefont {G\"uhne}}]{Nguyen2019LHSCriterion}%
  \BibitemOpen
  \bibfield  {author} {\bibinfo {author} {\bibfnamefont {H.~C.}\ \bibnamefont {Nguyen}}, \bibinfo {author} {\bibfnamefont {H.-V.}\ \bibnamefont {Nguyen}},\ and\ \bibinfo {author} {\bibfnamefont {O.}~\bibnamefont {G\"uhne}},\ }\bibfield  {title} {\bibinfo {title} {{Geometry of Einstein-Podolsky-Rosen Correlations}},\ }\href {https://doi.org/10.1103/PhysRevLett.122.240401} {\bibfield  {journal} {\bibinfo  {journal} {Phys. Rev. Lett.}\ }\textbf {\bibinfo {volume} {122}},\ \bibinfo {pages} {240401} (\bibinfo {year} {2019})},\ \Eprint {https://arxiv.org/abs/1808.09349}{arXiv:1808.09349 [quant-ph]}\BibitemShut {NoStop}%
\bibitem [{\citenamefont {{Ali}}\ \emph {et~al.}(2009)\citenamefont {{Ali}}, \citenamefont {{Carmeli}}, \citenamefont {{Heinosaari}},\ and\ \citenamefont {{Toigo}}}]{ali2009}%
  \BibitemOpen
  \bibfield  {author} {\bibinfo {author} {\bibfnamefont {S.~T.}\ \bibnamefont {{Ali}}}, \bibinfo {author} {\bibfnamefont {C.}~\bibnamefont {{Carmeli}}}, \bibinfo {author} {\bibfnamefont {T.}~\bibnamefont {{Heinosaari}}},\ and\ \bibinfo {author} {\bibfnamefont {A.}~\bibnamefont {{Toigo}}},\ }\bibfield  {title} {\bibinfo {title} {{Commutative POVMs and Fuzzy Observables}},\ }\href {https://doi.org/10.1007/s10701-009-9292-y} {\bibfield  {journal} {\bibinfo  {journal} {Foundations of Physics}\ }\textbf {\bibinfo {volume} {39}},\ \bibinfo {pages} {593--612} (\bibinfo {year} {2009})},\ \Eprint {https://arxiv.org/abs/0903.0523}{arXiv:0903.0523 [quant-ph]}\BibitemShut {NoStop}%
\bibitem [{\citenamefont {{Uola}}\ \emph {et~al.}(2015)\citenamefont {{Uola}}, \citenamefont {{Budroni}}, \citenamefont {{G{\"u}hne}},\ and\ \citenamefont {{Pellonp{\"a}{\"a}}}}]{uola15}%
  \BibitemOpen
  \bibfield  {author} {\bibinfo {author} {\bibfnamefont {R.}~\bibnamefont {{Uola}}}, \bibinfo {author} {\bibfnamefont {C.}~\bibnamefont {{Budroni}}}, \bibinfo {author} {\bibfnamefont {O.}~\bibnamefont {{G{\"u}hne}}},\ and\ \bibinfo {author} {\bibfnamefont {J.-P.}\ \bibnamefont {{Pellonp{\"a}{\"a}}}},\ }\bibfield  {title} {\bibinfo {title} {{One-to-One Mapping between Steering and Joint Measurability Problems}},\ }\href {https://doi.org/10.1103/PhysRevLett.115.230402} {\bibfield  {journal} {\bibinfo  {journal} {Phys. Rev. Lett.}\ }\textbf {\bibinfo {volume} {115}},\ \bibinfo {eid} {230402} (\bibinfo {year} {2015})},\ \Eprint {https://arxiv.org/abs/1507.08633}{arXiv:1507.08633 [quant-ph]}\BibitemShut {NoStop}%
\bibitem [{\citenamefont {{Fiorini}}\ \emph {et~al.}(2015)\citenamefont {{Fiorini}}, \citenamefont {{Massar}}, \citenamefont {{Pokutta}}, \citenamefont {{Tiwary}},\ and\ \citenamefont {{de Wolf}}}]{fmptw2011jour}%
  \BibitemOpen
  \bibfield  {author} {\bibinfo {author} {\bibfnamefont {S.}~\bibnamefont {{Fiorini}}}, \bibinfo {author} {\bibfnamefont {S.}~\bibnamefont {{Massar}}}, \bibinfo {author} {\bibfnamefont {S.}~\bibnamefont {{Pokutta}}}, \bibinfo {author} {\bibfnamefont {H.~R.}\ \bibnamefont {{Tiwary}}},\ and\ \bibinfo {author} {\bibfnamefont {R.}~\bibnamefont {{de Wolf}}},\ }\bibfield  {title} {\bibinfo {title} {{Exponential Lower Bounds for Polytopes in Combinatorial Optimization}},\ }\href {https://doi.org/10.1145/2716307} {\bibfield  {journal} {\bibinfo  {journal} {Journal of the ACM}\ }\textbf {\bibinfo {volume} {62}},\ \bibinfo {pages} {1--17} (\bibinfo {year} {2015})},\ \Eprint {https://arxiv.org/abs/1111.0837}{arXiv:1111.0837 [math-CO]}\BibitemShut {NoStop}%
\bibitem [{\citenamefont {Braun}\ \emph {et~al.}(2015)\citenamefont {Braun}, \citenamefont {Fiorini}, \citenamefont {Pokutta},\ and\ \citenamefont {Steurer}}]{bfps2012jour}%
  \BibitemOpen
  \bibfield  {author} {\bibinfo {author} {\bibfnamefont {G.}~\bibnamefont {Braun}}, \bibinfo {author} {\bibfnamefont {S.}~\bibnamefont {Fiorini}}, \bibinfo {author} {\bibfnamefont {S.}~\bibnamefont {Pokutta}},\ and\ \bibinfo {author} {\bibfnamefont {D.}~\bibnamefont {Steurer}},\ }\bibfield  {title} {\bibinfo {title} {{Approximation Limits of Linear Programs (Beyond Hierarchies)}},\ }\href {https://doi.org/10.1287/moor.2014.0694} {\bibfield  {journal} {\bibinfo  {journal} {{Math. Oper. Res.}}\ }\textbf {\bibinfo {volume} {40}},\ \bibinfo {pages} {179--199} (\bibinfo {year} {2015})},\ \Eprint {https://arxiv.org/abs/1204.0957}{arXiv:1204.0957 [cs.CC]}\BibitemShut {NoStop}%
\bibitem [{\citenamefont {Braun}\ and\ \citenamefont {Pokutta}(2014)}]{braun2014matching}%
  \BibitemOpen
  \bibfield  {author} {\bibinfo {author} {\bibfnamefont {G.}~\bibnamefont {Braun}}\ and\ \bibinfo {author} {\bibfnamefont {S.}~\bibnamefont {Pokutta}},\ }\bibfield  {title} {\bibinfo {title} {The matching polytope does not admit fully-polynomial size relaxation schemes},\ }\href {https://doi.org/10.1137/1.9781611973730.57} {\bibfield  {journal} {\bibinfo  {journal} {Proceedings of the Twenty-Sixth Annual ACM-SIAM Symposium on Discrete Algorithms}\ ,\ \bibinfo {pages} {837--846}} (\bibinfo {year} {2014})},\ \Eprint {https://arxiv.org/abs/1403.6710}{arXiv:1403.6710 [cs.CC]}\BibitemShut {NoStop}%
\bibitem [{\citenamefont {Gr{\"o}tschel}\ \emph {et~al.}(1981)\citenamefont {Gr{\"o}tschel}, \citenamefont {Lov{\'a}sz},\ and\ \citenamefont {Schrijver}}]{grotschel1981ellipsoid}%
  \BibitemOpen
  \bibfield  {author} {\bibinfo {author} {\bibfnamefont {M.}~\bibnamefont {Gr{\"o}tschel}}, \bibinfo {author} {\bibfnamefont {L.}~\bibnamefont {Lov{\'a}sz}},\ and\ \bibinfo {author} {\bibfnamefont {A.}~\bibnamefont {Schrijver}},\ }\bibfield  {title} {\bibinfo {title} {The ellipsoid method and its consequences in combinatorial optimization},\ }\href {https://doi.org/10.1007/BF02579273} {\bibfield  {journal} {\bibinfo  {journal} {Combinatorica}\ }\textbf {\bibinfo {volume} {1}},\ \bibinfo {pages} {169--197} (\bibinfo {year} {1981})}\BibitemShut {NoStop}%
\bibitem [{\citenamefont {Alizadeh}(1995)}]{alizadeh1995interior}%
  \BibitemOpen
  \bibfield  {author} {\bibinfo {author} {\bibfnamefont {F.}~\bibnamefont {Alizadeh}},\ }\bibfield  {title} {\bibinfo {title} {{Interior Point Methods in Semidefinite Programming with Applications to Combinatorial Optimization}},\ }\href {https://doi.org/10.1137/0805002} {\bibfield  {journal} {\bibinfo  {journal} {SIAM Journal on Optimization}\ }\textbf {\bibinfo {volume} {5}},\ \bibinfo {pages} {13--51} (\bibinfo {year} {1995})}\BibitemShut {NoStop}%
\bibitem [{\citenamefont {Nesterov}\ and\ \citenamefont {Nemirovskii}(1994)}]{nesterov1994interior}%
  \BibitemOpen
  \bibfield  {author} {\bibinfo {author} {\bibfnamefont {Y.}~\bibnamefont {Nesterov}}\ and\ \bibinfo {author} {\bibfnamefont {A.}~\bibnamefont {Nemirovskii}},\ }\href {https://doi.org/10.1137/1.9781611970791} {\bibinfo {title} {Interior-Point Polynomial Algorithms in Convex Programming}}\ (\bibinfo  {publisher} {Society for Industrial and Applied Mathematics},\ \bibinfo {year} {1994})\BibitemShut {NoStop}%
\bibitem [{\citenamefont {Vandenberghe}\ and\ \citenamefont {Boyd}(1996)}]{vandenberghe1996semidefinite}%
  \BibitemOpen
  \bibfield  {author} {\bibinfo {author} {\bibfnamefont {L.}~\bibnamefont {Vandenberghe}}\ and\ \bibinfo {author} {\bibfnamefont {S.}~\bibnamefont {Boyd}},\ }\bibfield  {title} {\bibinfo {title} {{Semidefinite Programming}},\ }\href {https://doi.org/10.1137/1038003} {\bibfield  {journal} {\bibinfo  {journal} {SIAM Review}\ }\textbf {\bibinfo {volume} {38}},\ \bibinfo {pages} {49--95} (\bibinfo {year} {1996})}\BibitemShut {NoStop}%
\bibitem [{\citenamefont {Aubrun}\ and\ \citenamefont {Szarek}(2017)}]{AS17}%
  \BibitemOpen
  \bibfield  {author} {\bibinfo {author} {\bibfnamefont {G.}~\bibnamefont {Aubrun}}\ and\ \bibinfo {author} {\bibfnamefont {S.~J.}\ \bibnamefont {Szarek}},\ }\href {https://doi.org/10.1090/surv/223} {\bibinfo {title} {{Alice and Bob Meet Banach: The Interface of Asymptotic Geometric Analysis and Quantum Information Theory}}},\ Mathematical Surveys and Monographs, Volume 223\ (\bibinfo  {publisher} {American Mathematical Society},\ \bibinfo {year} {2017})\BibitemShut {NoStop}%
\bibitem [{\citenamefont {Hirsch}\ \emph {et~al.}(2017)\citenamefont {Hirsch}, \citenamefont {Quintino}, \citenamefont {V{\'{e}}rtesi}, \citenamefont {Navascu{\'{e}}s},\ and\ \citenamefont {Brunner}}]{hirsch16}%
  \BibitemOpen
  \bibfield  {author} {\bibinfo {author} {\bibfnamefont {F.}~\bibnamefont {Hirsch}}, \bibinfo {author} {\bibfnamefont {M.~T.}\ \bibnamefont {Quintino}}, \bibinfo {author} {\bibfnamefont {T.}~\bibnamefont {V{\'{e}}rtesi}}, \bibinfo {author} {\bibfnamefont {M.}~\bibnamefont {Navascu{\'{e}}s}},\ and\ \bibinfo {author} {\bibfnamefont {N.}~\bibnamefont {Brunner}},\ }\bibfield  {title} {\bibinfo {title} {Better local hidden variable models for two-qubit {W}erner states and an upper bound on the {G}rothendieck constant {$K_G(3)$}},\ }\href {https://doi.org/10.22331/q-2017-04-25-3} {\bibfield  {journal} {\bibinfo  {journal} {{Quantum}}\ }\textbf {\bibinfo {volume} {1}},\ \bibinfo {pages} {3} (\bibinfo {year} {2017})},\ \Eprint {https://arxiv.org/abs/1609.06114}{arXiv:1609.06114 [quant-ph]}\BibitemShut {NoStop}%
\bibitem [{\citenamefont {Renegar}(2001)}]{renegar2001mathematical}%
  \BibitemOpen
  \bibfield  {author} {\bibinfo {author} {\bibfnamefont {J.}~\bibnamefont {Renegar}},\ }\href {https://doi.org/10.1137/1.9780898718812} {\bibinfo {title} {{A Mathematical View of Interior-Point Methods in Convex Optimization}}}\ (\bibinfo  {publisher} {Society for Industrial and Applied Mathematics},\ \bibinfo {year} {2001})\BibitemShut {NoStop}%
\bibitem [{\citenamefont {Boyd}\ and\ \citenamefont {Vandenberghe}(2004)}]{boyd2004convex}%
  \BibitemOpen
  \bibfield  {author} {\bibinfo {author} {\bibfnamefont {S.~P.}\ \bibnamefont {Boyd}}\ and\ \bibinfo {author} {\bibfnamefont {L.}~\bibnamefont {Vandenberghe}},\ }\href {https://doi.org/10.1017/CBO9780511804441} {\bibinfo {title} {{Convex Optimization}}}\ (\bibinfo  {publisher} {Cambridge University Press},\ \bibinfo {year} {2004})\BibitemShut {NoStop}%
\bibitem [{\citenamefont {Chiribella}\ and\ \citenamefont {Mauro~D’Ariano}(2006)}]{chiribella06}%
  \BibitemOpen
  \bibfield  {author} {\bibinfo {author} {\bibfnamefont {G.}~\bibnamefont {Chiribella}}\ and\ \bibinfo {author} {\bibfnamefont {G.}~\bibnamefont {Mauro~D’Ariano}},\ }\bibfield  {title} {\bibinfo {title} {Extremal covariant measurements},\ }\href {https://doi.org/10.1063/1.2349481} {\bibfield  {journal} {\bibinfo  {journal} {J. Math. Phys.}\ }\textbf {\bibinfo {volume} {47}},\ \bibinfo {pages} {092107} (\bibinfo {year} {2006})},\ \Eprint {https://arxiv.org/abs/quant-ph/0603168}{arXiv:quant-ph/0603168 [quant-ph]}\BibitemShut {NoStop}%
\bibitem [{\citenamefont {Renegar}(1988)}]{Renegar1988}%
  \BibitemOpen
  \bibfield  {author} {\bibinfo {author} {\bibfnamefont {J.}~\bibnamefont {Renegar}},\ }\bibfield  {title} {\bibinfo {title} {{A polynomial-time algorithm, based on Newton's method, for linear programming}},\ }\href {https://doi.org/10.1007/BF01580724} {\bibfield  {journal} {\bibinfo  {journal} {Math. Program.}\ }\textbf {\bibinfo {volume} {40}},\ \bibinfo {pages} {59–93} (\bibinfo {year} {1988})}\BibitemShut {NoStop}%
\bibitem [{\citenamefont {Wright}(1997)}]{Wright1997}%
  \BibitemOpen
  \bibfield  {author} {\bibinfo {author} {\bibfnamefont {S.~J.}\ \bibnamefont {Wright}},\ }\href {https://doi.org/10.1137/1.9781611971453} {\bibinfo {title} {Primal-Dual Interior-Point Methods}}\ (\bibinfo  {publisher} {Society for Industrial and Applied Mathematics},\ \bibinfo {year} {1997})\BibitemShut {NoStop}%
\bibitem [{\citenamefont {{Hardin}}\ \emph {et~al.}(1994, 2000)\citenamefont {{Hardin}}, \citenamefont {{Sloane}},\ and\ \citenamefont {{Smith}}}]{HardinSloane_spherical_codes}%
  \BibitemOpen
  \bibfield  {author} {\bibinfo {author} {\bibfnamefont {R.~H.}\ \bibnamefont {{Hardin}}}, \bibinfo {author} {\bibfnamefont {J.~A.}\ \bibnamefont {{Sloane}}},\ and\ \bibinfo {author} {\bibfnamefont {W.~D.}\ \bibnamefont {{Smith}}},\ }\href@noop {} {\bibinfo {title} {Tables of spherical codes with icosahedral symmetry, published electronically at \url{http://NeilSloane.com/icosahedral.codes/}}} (\bibinfo {year} {1994, 2000})\BibitemShut {NoStop}%
\bibitem [{\citenamefont {Lörwald}\ and\ \citenamefont {Reinelt}(2015)}]{LR15}%
  \BibitemOpen
  \bibfield  {author} {\bibinfo {author} {\bibfnamefont {S.}~\bibnamefont {Lörwald}}\ and\ \bibinfo {author} {\bibfnamefont {G.}~\bibnamefont {Reinelt}},\ }\bibfield  {title} {\bibinfo {title} {Panda: a software for polyhedral transformations},\ }\href {https://doi.org/10.1007/s13675-015-0040-0} {\bibfield  {journal} {\bibinfo  {journal} {EURO Journal on Computational Optimization}\ }\textbf {\bibinfo {volume} {3}},\ \bibinfo {pages} {297--308} (\bibinfo {year} {2015})}\BibitemShut {NoStop}%
\bibitem [{\citenamefont {{Skrzypczyk}}\ \emph {et~al.}(2020)\citenamefont {{Skrzypczyk}}, \citenamefont {{Hoban}}, \citenamefont {{Sainz}},\ and\ \citenamefont {{Linden}}}]{Skrzypczyk2020}%
  \BibitemOpen
  \bibfield  {author} {\bibinfo {author} {\bibfnamefont {P.}~\bibnamefont {{Skrzypczyk}}}, \bibinfo {author} {\bibfnamefont {M.~J.}\ \bibnamefont {{Hoban}}}, \bibinfo {author} {\bibfnamefont {A.~B.}\ \bibnamefont {{Sainz}}},\ and\ \bibinfo {author} {\bibfnamefont {N.}~\bibnamefont {{Linden}}},\ }\bibfield  {title} {\bibinfo {title} {{Complexity of compatible measurements}},\ }\href {https://doi.org/10.1103/PhysRevResearch.2.023292} {\bibfield  {journal} {\bibinfo  {journal} {Phys. Rev. Res.}\ }\textbf {\bibinfo {volume} {2}},\ \bibinfo {eid} {023292} (\bibinfo {year} {2020})},\ \Eprint {https://arxiv.org/abs/1908.10085}{arXiv:1908.10085 [quant-ph]}\BibitemShut {NoStop}%
\bibitem [{\citenamefont {{Ohst}}\ \emph {et~al.}(2024)\citenamefont {{Ohst}}, \citenamefont {{Yu}}, \citenamefont {{G{\"u}hne}},\ and\ \citenamefont {{Nguyen}}}]{ohst2024adaptive}%
  \BibitemOpen
  \bibfield  {author} {\bibinfo {author} {\bibfnamefont {T.-A.}\ \bibnamefont {{Ohst}}}, \bibinfo {author} {\bibfnamefont {X.-D.}\ \bibnamefont {{Yu}}}, \bibinfo {author} {\bibfnamefont {O.}~\bibnamefont {{G{\"u}hne}}},\ and\ \bibinfo {author} {\bibfnamefont {C.~H.}\ \bibnamefont {{Nguyen}}},\ }\bibfield  {title} {\bibinfo {title} {{Certifying Quantum Separability with Adaptive Polytopes}},\ }\href {https://doi.org/10.21468/SciPostPhys.16.3.063} {\bibfield  {journal} {\bibinfo  {journal} {SciPost Physics}\ }\textbf {\bibinfo {volume} {16}},\ \bibinfo {eid} {063} (\bibinfo {year} {2024})},\ \Eprint {https://arxiv.org/abs/2210.10054}{arXiv:2210.10054 [quant-ph]}\BibitemShut {NoStop}%
\bibitem [{\citenamefont {Porto}(2025)}]{gitcode}%
  \BibitemOpen
  \bibfield  {author} {\bibinfo {author} {\bibfnamefont {L.}~\bibnamefont {Porto}},\ }\href@noop {} {\bibinfo {title} {Github repository: Measurement incompatibility and quantum steering via linear programming}},\ \bibinfo {howpublished} {\url{https://github.com/lucporto/LpJm}} (\bibinfo {year} {2025})\BibitemShut {NoStop}%
\bibitem [{\citenamefont {{Gonz{\'a}lez}}(2010)}]{gonzalez2010fibonacci}%
  \BibitemOpen
  \bibfield  {author} {\bibinfo {author} {\bibfnamefont {{\'A}.}~\bibnamefont {{Gonz{\'a}lez}}},\ }\bibfield  {title} {\bibinfo {title} {{Measurement of Areas on a Sphere Using Fibonacci and Latitude-Longitude Lattices}},\ }\href {https://doi.org/10.1007/s11004-009-9257-x} {\bibfield  {journal} {\bibinfo  {journal} {Mathematical Geosciences}\ }\textbf {\bibinfo {volume} {42}},\ \bibinfo {pages} {49--64} (\bibinfo {year} {2010})},\ \Eprint {https://arxiv.org/abs/0912.4540}{arXiv:0912.4540 [math.MG]}\BibitemShut {NoStop}%
\bibitem [{\citenamefont {{Bavaresco}}\ \emph {et~al.}(2017)\citenamefont {{Bavaresco}}, \citenamefont {{Quintino}}, \citenamefont {{Guerini}}, \citenamefont {{Maciel}}, \citenamefont {{Cavalcanti}},\ and\ \citenamefont {{Cunha}}}]{bavaresco17}%
  \BibitemOpen
  \bibfield  {author} {\bibinfo {author} {\bibfnamefont {J.}~\bibnamefont {{Bavaresco}}}, \bibinfo {author} {\bibfnamefont {M.~T.}\ \bibnamefont {{Quintino}}}, \bibinfo {author} {\bibfnamefont {L.}~\bibnamefont {{Guerini}}}, \bibinfo {author} {\bibfnamefont {T.~O.}\ \bibnamefont {{Maciel}}}, \bibinfo {author} {\bibfnamefont {D.}~\bibnamefont {{Cavalcanti}}},\ and\ \bibinfo {author} {\bibfnamefont {M.~T.}\ \bibnamefont {{Cunha}}},\ }\bibfield  {title} {\bibinfo {title} {{Most incompatible measurements for robust steering tests}},\ }\href {https://doi.org/10.1103/PhysRevA.96.022110} {\bibfield  {journal} {\bibinfo  {journal} {Phys. Rev. A}\ }\textbf {\bibinfo {volume} {96}},\ \bibinfo {eid} {022110} (\bibinfo {year} {2017})},\ \Eprint {https://arxiv.org/abs/1704.02994}{arXiv:1704.02994 [quant-ph]}\BibitemShut {NoStop}%
\bibitem [{\citenamefont {Werner}(1989)}]{werner89}%
  \BibitemOpen
  \bibfield  {author} {\bibinfo {author} {\bibfnamefont {R.~F.}\ \bibnamefont {Werner}},\ }\bibfield  {title} {\bibinfo {title} {Quantum states with {Einstein}-{Podolsky}-{Rosen} correlations admitting a hidden-variable model},\ }\href {https://doi.org/10.1103/PhysRevA.40.4277} {\bibfield  {journal} {\bibinfo  {journal} {Phys. Rev. A}\ }\textbf {\bibinfo {volume} {40}},\ \bibinfo {pages} {4277--4281} (\bibinfo {year} {1989})}\BibitemShut {NoStop}%
\bibitem [{\citenamefont {{Zhang}}\ and\ \citenamefont {{Chitambar}}(2024)}]{zhang2024}%
  \BibitemOpen
  \bibfield  {author} {\bibinfo {author} {\bibfnamefont {Y.}~\bibnamefont {{Zhang}}}\ and\ \bibinfo {author} {\bibfnamefont {E.}~\bibnamefont {{Chitambar}}},\ }\bibfield  {title} {\bibinfo {title} {{Exact Steering Bound for Two-Qubit Werner States}},\ }\href {https://doi.org/10.1103/PhysRevLett.132.250201} {\bibfield  {journal} {\bibinfo  {journal} {Phys. Rev. Lett.}\ }\textbf {\bibinfo {volume} {132}},\ \bibinfo {eid} {250201} (\bibinfo {year} {2024})},\ \Eprint {https://arxiv.org/abs/2309.09960}{arXiv:2309.09960 [quant-ph]}\BibitemShut {NoStop}%
\bibitem [{\citenamefont {{Renner}}(2024)}]{renner2024}%
  \BibitemOpen
  \bibfield  {author} {\bibinfo {author} {\bibfnamefont {M.~J.}\ \bibnamefont {{Renner}}},\ }\bibfield  {title} {\bibinfo {title} {{Compatibility of Generalized Noisy Qubit Measurements}},\ }\href {https://doi.org/10.1103/PhysRevLett.132.250202} {\bibfield  {journal} {\bibinfo  {journal} {Phys. Rev. Lett.}\ }\textbf {\bibinfo {volume} {132}},\ \bibinfo {eid} {250202} (\bibinfo {year} {2024})},\ \Eprint {https://arxiv.org/abs/2309.12290}{arXiv:2309.12290 [quant-ph]}\BibitemShut {NoStop}%
\bibitem [{\citenamefont {{MOSEK ApS}}(2025)}]{MOSEK}%
  \BibitemOpen
  \bibfield  {author} {\bibinfo {author} {\bibnamefont {{MOSEK ApS}}},\ }\href {https://docs.mosek.com/latest/juliaapi/index.html} {\bibinfo {title} {MOSEK Optimizer API for Julia 11.0.20}} (\bibinfo {year} {2025})\BibitemShut {NoStop}%
\bibitem [{\citenamefont {{Andrejic}}\ and\ \citenamefont {{Kunjwal}}(2020)}]{andrejic2020}%
  \BibitemOpen
  \bibfield  {author} {\bibinfo {author} {\bibfnamefont {N.}~\bibnamefont {{Andrejic}}}\ and\ \bibinfo {author} {\bibfnamefont {R.}~\bibnamefont {{Kunjwal}}},\ }\bibfield  {title} {\bibinfo {title} {{Joint measurability structures realizable with qubit measurements: Incompatibility via marginal surgery}},\ }\href {https://doi.org/10.1103/PhysRevResearch.2.043147} {\bibfield  {journal} {\bibinfo  {journal} {Physical Review Research}\ }\textbf {\bibinfo {volume} {2}},\ \bibinfo {eid} {043147} (\bibinfo {year} {2020})},\ \Eprint {https://arxiv.org/abs/2003.00785}{arXiv:2003.00785 [quant-ph]}\BibitemShut {NoStop}%
\bibitem [{\citenamefont {Araújo}\ \emph {et~al.}(2025)\citenamefont {Araújo}, \citenamefont {Brown}, \citenamefont {Designolle}, \citenamefont {de~Gois}, \citenamefont {Liu}, \citenamefont {Porto},\ and\ \citenamefont {Quintino}}]{Ket.jl}%
  \BibitemOpen
  \bibfield  {author} {\bibinfo {author} {\bibfnamefont {M.}~\bibnamefont {Araújo}}, \bibinfo {author} {\bibfnamefont {P.}~\bibnamefont {Brown}}, \bibinfo {author} {\bibfnamefont {S.}~\bibnamefont {Designolle}}, \bibinfo {author} {\bibfnamefont {C.}~\bibnamefont {de~Gois}}, \bibinfo {author} {\bibfnamefont {Y.-C.}\ \bibnamefont {Liu}}, \bibinfo {author} {\bibfnamefont {L.}~\bibnamefont {Porto}},\ and\ \bibinfo {author} {\bibfnamefont {M.~T.}\ \bibnamefont {Quintino}},\ }\href {https://doi.org/10.5281/zenodo.14674642} {\bibinfo {title} {{Ket.jl: a Julia toolbox for quantum information, nonlocality, and entanglement}}} (\bibinfo {year} {2025})\BibitemShut {NoStop}%
\bibitem [{\citenamefont {{Heinosaari}}\ \emph {et~al.}(2020)\citenamefont {{Heinosaari}}, \citenamefont {{Jivulescu}},\ and\ \citenamefont {{Nechita}}}]{heinossari2020random}%
  \BibitemOpen
  \bibfield  {author} {\bibinfo {author} {\bibfnamefont {T.}~\bibnamefont {{Heinosaari}}}, \bibinfo {author} {\bibfnamefont {M.~A.}\ \bibnamefont {{Jivulescu}}},\ and\ \bibinfo {author} {\bibfnamefont {I.}~\bibnamefont {{Nechita}}},\ }\bibfield  {title} {\bibinfo {title} {{Random positive operator valued measures}},\ }\href {https://doi.org/10.1063/1.5131028} {\bibfield  {journal} {\bibinfo  {journal} {J. Math. Phys.}\ }\textbf {\bibinfo {volume} {61}},\ \bibinfo {eid} {042202} (\bibinfo {year} {2020})},\ \Eprint {https://arxiv.org/abs/1902.04751}{arXiv:1902.04751 [quant-ph]}\BibitemShut {NoStop}%
\bibitem [{\citenamefont {{Mauro D'Ariano}}\ \emph {et~al.}(2005)\citenamefont {{Mauro D'Ariano}}, \citenamefont {{Lo Presti}},\ and\ \citenamefont {{Perinotti}}}]{dariano05}%
  \BibitemOpen
  \bibfield  {author} {\bibinfo {author} {\bibfnamefont {G.}~\bibnamefont {{Mauro D'Ariano}}}, \bibinfo {author} {\bibfnamefont {P.}~\bibnamefont {{Lo Presti}}},\ and\ \bibinfo {author} {\bibfnamefont {P.}~\bibnamefont {{Perinotti}}},\ }\bibfield  {title} {\bibinfo {title} {{Classical randomness in quantum measurements}},\ }\href {https://doi.org/10.1088/0305-4470/38/26/010} {\bibfield  {journal} {\bibinfo  {journal} {J. Phys. A: Math. Gen.}\ }\textbf {\bibinfo {volume} {38}},\ \bibinfo {pages} {5979--5991} (\bibinfo {year} {2005})},\ \Eprint {https://arxiv.org/abs/quant-ph/0408115}{arXiv:quant-ph/0408115}\BibitemShut {NoStop}%
\bibitem [{\citenamefont {{Saunders}}\ \emph {et~al.}(2010)\citenamefont {{Saunders}}, \citenamefont {{Jones}}, \citenamefont {{Wiseman}},\ and\ \citenamefont {{Pryde}}}]{saunders10}%
  \BibitemOpen
  \bibfield  {author} {\bibinfo {author} {\bibfnamefont {D.~J.}\ \bibnamefont {{Saunders}}}, \bibinfo {author} {\bibfnamefont {S.~J.}\ \bibnamefont {{Jones}}}, \bibinfo {author} {\bibfnamefont {H.~M.}\ \bibnamefont {{Wiseman}}},\ and\ \bibinfo {author} {\bibfnamefont {G.~J.}\ \bibnamefont {{Pryde}}},\ }\bibfield  {title} {\bibinfo {title} {{Experimental EPR-steering using Bell-local states}},\ }\href {https://doi.org/10.1038/nphys1766} {\bibfield  {journal} {\bibinfo  {journal} {Nature Physics}\ }\textbf {\bibinfo {volume} {6}},\ \bibinfo {pages} {845--849} (\bibinfo {year} {2010})},\ \Eprint {https://arxiv.org/abs/0909.0805}{arXiv:0909.0805 [quant-ph]}\BibitemShut {NoStop}%
\bibitem [{\citenamefont {{Cavalcanti}}\ \emph {et~al.}(2009)\citenamefont {{Cavalcanti}}, \citenamefont {{Jones}}, \citenamefont {{Wiseman}},\ and\ \citenamefont {{Reid}}}]{cavalcanti09}%
  \BibitemOpen
  \bibfield  {author} {\bibinfo {author} {\bibfnamefont {E.~G.}\ \bibnamefont {{Cavalcanti}}}, \bibinfo {author} {\bibfnamefont {S.~J.}\ \bibnamefont {{Jones}}}, \bibinfo {author} {\bibfnamefont {H.~M.}\ \bibnamefont {{Wiseman}}},\ and\ \bibinfo {author} {\bibfnamefont {M.~D.}\ \bibnamefont {{Reid}}},\ }\bibfield  {title} {\bibinfo {title} {{Experimental criteria for steering and the Einstein-Podolsky-Rosen paradox}},\ }\href {https://doi.org/10.1103/PhysRevA.80.032112} {\bibfield  {journal} {\bibinfo  {journal} {Phys. Rev. A}\ }\textbf {\bibinfo {volume} {80}},\ \bibinfo {eid} {032112} (\bibinfo {year} {2009})},\ \Eprint {https://arxiv.org/abs/0907.1109}{arXiv:0907.1109 [quant-ph]}\BibitemShut {NoStop}%
\bibitem [{\citenamefont {Quintino}\ \emph {et~al.}(2015)\citenamefont {Quintino}, \citenamefont {V\'ertesi}, \citenamefont {Cavalcanti}, \citenamefont {Augusiak}, \citenamefont {Demianowicz}, \citenamefont {Ac\'{\i}n},\ and\ \citenamefont {Brunner}}]{quintino15}%
  \BibitemOpen
  \bibfield  {author} {\bibinfo {author} {\bibfnamefont {M.~T.}\ \bibnamefont {Quintino}}, \bibinfo {author} {\bibfnamefont {T.}~\bibnamefont {V\'ertesi}}, \bibinfo {author} {\bibfnamefont {D.}~\bibnamefont {Cavalcanti}}, \bibinfo {author} {\bibfnamefont {R.}~\bibnamefont {Augusiak}}, \bibinfo {author} {\bibfnamefont {M.}~\bibnamefont {Demianowicz}}, \bibinfo {author} {\bibfnamefont {A.}~\bibnamefont {Ac\'{\i}n}},\ and\ \bibinfo {author} {\bibfnamefont {N.}~\bibnamefont {Brunner}},\ }\bibfield  {title} {\bibinfo {title} {Inequivalence of entanglement, steering, and {Bell} nonlocality for general measurements},\ }\href {https://doi.org/10.1103/PhysRevA.92.032107} {\bibfield  {journal} {\bibinfo  {journal} {Phys. Rev. A}\ }\textbf {\bibinfo {volume} {92}},\ \bibinfo {pages} {032107} (\bibinfo {year} {2015})},\ \Eprint {https://arxiv.org/abs/1501.03332}{arXiv:1501.03332 [quant-ph]}\BibitemShut {NoStop}%
\bibitem [{\citenamefont {{Barrett}}(2002)}]{barrett02}%
  \BibitemOpen
  \bibfield  {author} {\bibinfo {author} {\bibfnamefont {J.}~\bibnamefont {{Barrett}}},\ }\bibfield  {title} {\bibinfo {title} {{Nonsequential positive-operator-valued measurements on entangled mixed states do not always violate a Bell inequality}},\ }\href {https://doi.org/10.1103/PhysRevA.65.042302} {\bibfield  {journal} {\bibinfo  {journal} {Phys. Rev. A}\ }\textbf {\bibinfo {volume} {65}},\ \bibinfo {eid} {042302} (\bibinfo {year} {2002})},\ \Eprint {https://arxiv.org/abs/quant-ph/0107045}{arXiv:quant-ph/0107045}\BibitemShut {NoStop}%
\bibitem [{\citenamefont {{Chau Nguyen}}\ \emph {et~al.}(2018)\citenamefont {{Chau Nguyen}}, \citenamefont {{Milne}}, \citenamefont {{Vu}},\ and\ \citenamefont {{Jevtic}}}]{nguyen2018}%
  \BibitemOpen
  \bibfield  {author} {\bibinfo {author} {\bibfnamefont {H.}~\bibnamefont {{Chau Nguyen}}}, \bibinfo {author} {\bibfnamefont {A.}~\bibnamefont {{Milne}}}, \bibinfo {author} {\bibfnamefont {T.}~\bibnamefont {{Vu}}},\ and\ \bibinfo {author} {\bibfnamefont {S.}~\bibnamefont {{Jevtic}}},\ }\bibfield  {title} {\bibinfo {title} {{Quantum steering with positive operator valued measures}},\ }\href {https://doi.org/10.1088/1751-8121/aad115} {\bibfield  {journal} {\bibinfo  {journal} {J. Phys. A: Math. Gen.}\ }\textbf {\bibinfo {volume} {51}},\ \bibinfo {eid} {355302} (\bibinfo {year} {2018})},\ \Eprint {https://arxiv.org/abs/1706.08166}{arXiv:1706.08166 [quant-ph]}\BibitemShut {NoStop}%
\bibitem [{\citenamefont {{Designolle}}\ \emph {et~al.}(2023)\citenamefont {{Designolle}}, \citenamefont {{Iommazzo}}, \citenamefont {{Besan{\c{c}}on}}, \citenamefont {{Knebel}}, \citenamefont {{Gel{\ss}}},\ and\ \citenamefont {{Pokutta}}}]{Designolle23Grothendieck}%
  \BibitemOpen
  \bibfield  {author} {\bibinfo {author} {\bibfnamefont {S.}~\bibnamefont {{Designolle}}}, \bibinfo {author} {\bibfnamefont {G.}~\bibnamefont {{Iommazzo}}}, \bibinfo {author} {\bibfnamefont {M.}~\bibnamefont {{Besan{\c{c}}on}}}, \bibinfo {author} {\bibfnamefont {S.}~\bibnamefont {{Knebel}}}, \bibinfo {author} {\bibfnamefont {P.}~\bibnamefont {{Gel{\ss}}}},\ and\ \bibinfo {author} {\bibfnamefont {S.}~\bibnamefont {{Pokutta}}},\ }\bibfield  {title} {\bibinfo {title} {{Improved local models and new Bell inequalities via Frank-Wolfe algorithms}},\ }\href {https://doi.org/10.1103/PhysRevResearch.5.043059} {\bibfield  {journal} {\bibinfo  {journal} {Phys. Rev. Res.}\ }\textbf {\bibinfo {volume} {5}},\ \bibinfo {eid} {043059} (\bibinfo {year} {2023})},\ \Eprint {https://arxiv.org/abs/2302.04721}{arXiv:2302.04721 [quant-ph]}\BibitemShut {NoStop}%
\bibitem [{\citenamefont {Designolle}(2025)}]{gitdata}%
  \BibitemOpen
  \bibfield  {author} {\bibinfo {author} {\bibfnamefont {S.}~\bibnamefont {Designolle}},\ }\href@noop {} {\bibinfo {title} {Github repository: Approximation of the set of quantum states via polytopes}},\ \bibinfo {howpublished} {\url{https://github.com/sebastiendesignolle/ApproximationQuantumStates}} (\bibinfo {year} {2025})\BibitemShut {NoStop}%
\bibitem [{\citenamefont {Hughes}\ and\ \citenamefont {Waldron}(2021)}]{hugues2021spherical}%
  \BibitemOpen
  \bibfield  {author} {\bibinfo {author} {\bibfnamefont {D.}~\bibnamefont {Hughes}}\ and\ \bibinfo {author} {\bibfnamefont {S.}~\bibnamefont {Waldron}},\ }\bibfield  {title} {\bibinfo {title} {Spherical (t,t)-designs with a small number of vectors},\ }\href {https://doi.org/https://doi.org/10.1016/j.laa.2020.08.010} {\bibfield  {journal} {\bibinfo  {journal} {Linear Algebra and its Applications}\ }\textbf {\bibinfo {volume} {608}},\ \bibinfo {pages} {84--106} (\bibinfo {year} {2021})}\BibitemShut {NoStop}%
\bibitem [{\citenamefont {Nguyen}\ \emph {et~al.}(2020)\citenamefont {Nguyen}, \citenamefont {Designolle}, \citenamefont {Barakat},\ and\ \citenamefont {Gühne}}]{NDBG20}%
  \BibitemOpen
  \bibfield  {author} {\bibinfo {author} {\bibfnamefont {H.~C.}\ \bibnamefont {Nguyen}}, \bibinfo {author} {\bibfnamefont {S.}~\bibnamefont {Designolle}}, \bibinfo {author} {\bibfnamefont {M.}~\bibnamefont {Barakat}},\ and\ \bibinfo {author} {\bibfnamefont {O.}~\bibnamefont {Gühne}},\ }\bibfield  {title} {\bibinfo {title} {Symmetries between measurements in quantum mechanics},\ }\href {https://arxiv.org/abs/2003.12553} {\bibfield  {journal} {\bibinfo  {journal} {arXiv:2003.12553}\ } (\bibinfo {year} {2020})}\BibitemShut {NoStop}%
\bibitem [{\citenamefont {Wootters}\ and\ \citenamefont {Fields}(1989)}]{WF89}%
  \BibitemOpen
  \bibfield  {author} {\bibinfo {author} {\bibfnamefont {W.~K.}\ \bibnamefont {Wootters}}\ and\ \bibinfo {author} {\bibfnamefont {B.~D.}\ \bibnamefont {Fields}},\ }\bibfield  {title} {\bibinfo {title} {Optimal state-determination by mutually unbiased measurements},\ }\href {https://doi.org/https://doi.org/10.1016/0003-4916(89)90322-9} {\bibfield  {journal} {\bibinfo  {journal} {Annals of Physics}\ }\textbf {\bibinfo {volume} {191}},\ \bibinfo {pages} {363--381} (\bibinfo {year} {1989})}\BibitemShut {NoStop}%
\bibitem [{\citenamefont {McNulty}\ and\ \citenamefont {Weigert}(2024)}]{MW24}%
  \BibitemOpen
  \bibfield  {author} {\bibinfo {author} {\bibfnamefont {D.}~\bibnamefont {McNulty}}\ and\ \bibinfo {author} {\bibfnamefont {S.}~\bibnamefont {Weigert}},\ }\bibfield  {title} {\bibinfo {title} {{Mutually Unbiased Bases in Composite Dimensions -- A Review}},\ }\href {https://arxiv.org/abs/2410.23997} {\bibfield  {journal} {\bibinfo  {journal} {arXiv:2410.23997}\ } (\bibinfo {year} {2024})}\BibitemShut {NoStop}%
\bibitem [{\citenamefont {Bengtsson}\ and\ \citenamefont {Ericsson}(2005)}]{Bengtsson2005}%
  \BibitemOpen
  \bibfield  {author} {\bibinfo {author} {\bibfnamefont {I.}~\bibnamefont {Bengtsson}}\ and\ \bibinfo {author} {\bibfnamefont {{\AA}.}~\bibnamefont {Ericsson}},\ }\bibfield  {title} {\bibinfo {title} {{Mutually Unbiased Bases and the Complementarity Polytope}},\ }\href {https://doi.org/10.1007/s11080-005-5721-3} {\bibfield  {journal} {\bibinfo  {journal} {Open Systems {\&} Information Dynamics}\ }\textbf {\bibinfo {volume} {12}},\ \bibinfo {pages} {107--120} (\bibinfo {year} {2005})}\BibitemShut {NoStop}%
\end{thebibliography}%

\appendix

\section{Shrinking factors of highly symmetric sets of qutrit states}
\label{app:shr_sym_qutrit}

The method exposed in \cref{sec:polytopes} to obtain shrinking factors relies on rational polytopes, which severely impairs the elegance of the solutions, as this small denominator constraint restricts the overall symmetry.
Here we expose a somehow more natural construction, which is unfortunately limited by our capacity to enumerate facets of the resulting polytope.
When we are able to do so, we see that the sets of qutrit states crafted this way achieve remarkable shrinking factors (proven up to numerical precision) given their small number of points.
We believe they could be of independent interest and provide their coordinates to facilitate their usage~\cite{gitdata}.

The technique starts from a symmetric structure, for instance, a complex projective design as classified in~\cite{hugues2021spherical}, and enumerates the facets of the resulting polytope.
As explained around \cref{eq:shr}, the shrinking factor is limited by the few facets reaching the minimum.
A natural way of extending the polytope to increase its shrinking factor is then to add the projectors onto the maximum eigenvalues of the facets saturating the minimum in \cref{eq:shr}, which is in general a strict subset of all facets.
Iteratively repeating this procedure leads to bigger and bigger structures with an increasing shrinking factor.
However, the method obviously generates sooner of later polytopes for which the facet enumeration is intractable.

In \cref{tab:symmetric_qutrit}, we give a few families of polytopes obtained following this technique.
The first two lines exhibit some kind of duality, but if the polytope with 21 vertices indeed has 24 facets whose eigenprojectors give rise to the other polytope, this one has 2268 facets, only 21 of which reach the minimum in \cref{eq:shr}.
Two of the ST27 orbits also enjoy this property.
Note also that, out of the 81 facets of the polytope formed by MUBs (see \cref{sec:polytopes}), only 9 achieve the minimum in \cref{eq:shr}, and the corresponding eigenprojectors happen to form a symmetric informationally complete (SIC) POVM.

\begin{table}[ht]
  \centering
  \scalebox{0.9}{
    \begin{tabular}{|c|cc|cc|cc|cc|cc|cc|c|}
      \hline
      ST & \multicolumn{2}{c|}{Step 0} & \multicolumn{2}{c|}{Step 1} & \multicolumn{2}{c|}{Step 2} & \multicolumn{2}{c|}{Step 3} & \multicolumn{2}{c|}{Step 4} & \multicolumn{2}{c|}{Step 5} & Comments \\ \hline
      \multirow{3}{*}{24}
      & 21 & 0.4129 & 45  & 0.5323 & 129 & 0.6991 & 465 & ?      &     &        &     &   & \cite{hugues2021spherical} \\
      & 24 & 0.4960 & 45  & 0.5323 & 129 & 0.6991 & 465 & ?      &     &        &     &   &                            \\
      & 28 & 0.3627 & 52  & 0.5    & 73  & 0.6989 & 409 & ?      &     &        &     &   & \cite{hugues2021spherical} \\ \hline
      \multirow{3}{*}{25}
      & 12 & 0.25   & 21  & 0.4766 & 75  & 0.6543 & 147 & 0.7541 & 291 & 0.8177 & 507 & ? & MUBs                       \\ 
      & 36 & 0.5    & 48  & 0.6168 & 120 & 0.625  & 129 & 0.7721 & 345 & 0.8043 & 453 & ? & \cite{hugues2021spherical} \\ 
      & 72 & 0.4698 & 81  & 0.5686 & 93  & 0.7158 & 165 & 0.7761 & 381 & ?      &     &   &                            \\ \hline
      \multirow{3}{*}{27}
      & 36 & 0.5854 & 81  & 0.6920 & 153 & 0.7820 & 513 & ?      &     &        &     &   & \cite{hugues2021spherical} \\
      & 45 & 0.5854 & 81  & 0.6920 & 153 & 0.7820 & 513 & ?      &     &        &     &   & \cite{hugues2021spherical} \\
      & 60 & 0.5723 & 132 & 0.7217 & 492 & ?      &     &        &     &        &     &   & \cite{hugues2021spherical} \\ \hline
    \end{tabular}
  }
  \caption{
    Cardinality and shrinking factor of sets of 3-dimensional pure states constructed via the iterative procedure described in the main text starting from orbits of complex reflection groups indicated with their Shephard-Todd (ST) number.
    Most initial structures (step 0) are listed in~\cite{hugues2021spherical} but, using the code from~\cite{NDBG20}, two new ones appeared.
    The shrinking factor was computed up to some numerical precision, in contrast with the exact computations of panda described in \cref{sec:polytopes}.
    Question marks indicate polytope for which we could not enumerate all facets.
    Higher precision values as well as coordinates can be found in the corresponding repository~\cite{gitdata}.
  }
  \label{tab:symmetric_qutrit}
\end{table}

\section{Shrinking factors of complete sets of mutually unbiased bases}
\label{app:mubs}

The methods to construct polytopes with good shrinking factors discussed in \cref{sec:polytopes} and \cref{app:shr_sym_qutrit} are mostly limited by computational reasons.
Here we present a construction which is entirely analytical, so that such limitations are avoided.
Specifically, we consider the $d(d+1)$ states in a complete set of $d+1$ mutually unbiased bases (MUBs) in dimension $d$.
The dimensions where they are known to exist are the powers of prime numbers~\cite{WF89}, and it is still unknown whether they can be found in any other dimension~\cite{MW24}.
When such states, denoted $\ket{\varphi_a^x}$ for $a\in\{1,\ldots,d\}$ and $x\in\{1,\ldots,d+1\}$, exist, we claim that their shrinking factor is
\begin{equation}
  r=\frac{1}{d(1-\nu)+1},\quad\text{where}\quad\nu=\min_{\vec{j}\in\{1,\ldots,d\}^{d+1}}\left\{\min\Sp\left(\sum_x\ketbra{\varphi_{j_x}^x}\right)\right\}.
  \label{eq:shr_mub}
\end{equation}
Importantly, when $d$ is odd, for the standard construction of complete sets of MUBs~\cite{WF89}, the equality $\nu=0$ is known from~\cite[Appendix~E.3.4]{designolle2019quantifiers}, so that the shrinking factor reduces to $1/(d+1)$ in this case, generalising the value $\frac14$ observed in~\cref{tab:symmetric_qutrit}.

The proof of \cref{eq:shr_mub} follows from~\cite{Bengtsson2005}, where the facets of the so-called complementarity polytope with vertices $\ketbra{\varphi_a^x}$ are shown to be the $d^{d+1}$ inequalities
\begin{equation}
  \Tr\left(F_{\vec{j}}\,\rho\right)\leq-1,\quad\text{where}\quad F_{\vec{j}}=-\sum_x\ketbra{\varphi_{j_x}^x}.
\end{equation}
Noting that $\Tr F_{\vec{j}}=-(d+1)$ and $\max_{\vec{j}}\Sp F_{\vec{j}}=-\nu$, we can use \cref{eq:shr} to derive the expression given in \cref{eq:shr_mub}.

\end{document}